\documentclass[11pt]{article}
\usepackage[utf8]{inputenc}  
\usepackage[T1]{fontenc} 
\usepackage{lmodern}
\usepackage{amsmath} 
\usepackage{amssymb}
\usepackage{amsthm} 
\usepackage{mathrsfs}
\usepackage{graphicx}
\usepackage{ccaption}
\usepackage{hyperref}
\usepackage{listings}
\usepackage{enumitem}
\usepackage{tikz}
\usepackage{float}
\usepackage{pgfplots}
\pgfplotsset{compat=newest}
\usetikzlibrary{shapes}
\usetikzlibrary{patterns}
\usepackage{stmaryrd}
\usepackage{authblk}
\usepackage[pagewise]{lineno}
\usepackage{color}

\usepackage[left=2.5cm,right=2.5cm,top=3cm,bottom=3cm]{geometry}

\theoremstyle{definition}
\newtheorem{defi}{Definition}
\newtheorem{theorem}{Theorem}
\newtheorem{prop}{Proposition}
\newtheorem{asum}{Assumption}
\newtheorem{coro}{Corollary}
\newtheorem{lem}{Lemma}
\newtheorem{rem}{Remark}
\newcommand{\R}{\mathbb{R}}
\newcommand{\N}{\mathbb{N}}

\newcommand{\C}{\mathbb{C}}

\newcommand{\dd}{\mathrm{d}}

\newcommand{\eps}{\varepsilon}
\newcommand{\loc}{\text{loc}}

\newcommand{\ex}{\text{ex}}

\newcommand{\Ai}{\mathcal{A}}
\newcommand{\Bi}{\mathcal{B}}

    \pgfplotsset{
        colormap={parula}{
            rgb255=(53,42,135)
            rgb255=(15,92,221)
            rgb255=(18,125,216)
            rgb255=(7,156,207)
            rgb255=(21,177,180)
            rgb255=(89,189,140)
            rgb255=(165,190,107)
            rgb255=(225,185,82)
            rgb255=(252,206,46)
            rgb255=(249,251,14)
        },
    }
\everymath{\displaystyle}

\newcommand{\g}[1]{\boldsymbol{#1}}

\def\ph{\varphi}
\def\top{\text{top}}
\def\bot{\text{bot}}
\def\app{\text{app}}
\newcommand{\cC}{\mathcal{C}}
\def\cH{\mathcal{H}}
\def\cN{\mathcal{N}}
\def\cO{\mathcal{O}}
\def\tH{\text{H}}
\def\tL{\text{L}}
\DeclareMathOperator\supp{\text{supp}}
\newcommand{\norm}[2]{\left\|{#1}\right\|_{#2}}

\title{Reconstruction of smooth shape defects in waveguides using locally resonant frequencies}
\author[1,*]{Angèle Niclas}
\author[1]{Laurent Seppecher}
\affil[1]{Institut Camille Jordan, \'Ecole Centrale Lyon, France}
\affil[*]{Corresponding author: angele.niclas@ec-lyon.fr}
\date{}

\begin{document}
\maketitle 

\begin{abstract}
This article aims to present a new method to reconstruct slowly varying width defects in 2D waveguides using locally resonant frequencies. At these frequencies, locally resonant modes propagate in the waveguide under the form of Airy functions depending on a parameter called the locally resonant point. In this particular point, the local width of the waveguide is known and its location can be recovered from boundary measurements of the wavefield. Using the same process for different frequencies, we produce a good approximation of the width in all the waveguide. Given multi-frequency measurements taken at the surface of the waveguide, we provide a $\text{L}^\infty$-stable explicit method to reconstruct the width of the waveguide. We finally validate our method on numerical data, and we discuss its applications and limits. 
\end{abstract}

\section{Introduction}

This article presents a new method to reconstruct width variations of a slowly varying waveguide from multi-frequency one side boundary measurements in dimension 2. The considered varying waveguide is described by
\begin{equation} \Omega:=\left\{(x,y)\in \R^2 \, |\, 0<y<h(x)\right\},
\end{equation}
where $h\in \mathcal{C}^2(\R)\cap W^{2,\infty}(\R)$ is a positive profile function defining the top boundary. The bottom boundary is assumed to be flat (see an illustration in Figure \ref{meas}) but a similar analysis could be done when both boundaries are varying. In the time-harmonic regime, the wavefield $u$ satisfies the Helmholtz equation with Neumann boundary conditions
\begin{equation}\label{eqdebut}
\left\{\begin{array}{cl}
\Delta u+k^2u =- f & \text{ in } \Omega,\\ \partial_\nu u =b & \text{ on }\partial\Omega,
\end{array}\right.
\end{equation}
where $k\in (0,+\infty)$ is the frequency, $f$ is an interior source term, and $b$ is a boundary source term. In this work, a waveguide is said to be slowly varying when there exists a small parameter $\eta>0$ such that $\Vert h'\Vert_{\text{L}^\infty(\R)}\leq \eta$ and $\Vert h''\Vert_{\text{L}^\infty(\R)}\leq \eta^2$. Such waveguides are good models of ducts, corroded pipes, or metal plates (see \cite{honarvar1,legrand1}). 

We focus in this work on the recovery of the function $h$ modeling the waveguide shape from the knowledge of the wavefield $d^\text{ex}(x) :=u(x,0)$ on one surface of the waveguide and for multiple frequencies $k$. This model and inverse problem is inspired from non destructive monitoring of plates done in \cite{balogun1,legrand2,ces1}. Hence we assume the knowledge of measurements of $u(x,y)$ for $x\in I$ where $I$ is an interval of $\R$ and $y=0$ in a frequency interval $K\subset (0,+\infty)$, as shown in Figure \ref{meas}. 

If $k$ is chosen such that $k=n\pi/h(x^\star_k)$ with $n\in \N$ and $x^\star_k\in \R$, the Helmholtz problem is not well posed in general. Nevertheless, we prove in \cite{bonnetier2} that there exists a unique solution to this problem as long as the waveguide is slowly varying. In the same work, we also give a suitable explicit approximation of the wavefield that explicitly depends on $x^\star_k$. The aim in this article to recover the position of $x^\star_k$ for different frequencies, and then to recover the shape function $h$. Using these frequencies, the proposed inverse problem is highly non linear but a unique and stable recovering of $h$ is possible up to a controllable approximation error.

\begin{figure}[h]
\begin{center}
\begin{tikzpicture} 
\draw (-2,0) -- (9,0);
\draw (-2,1) -- (1,1); 
\draw (7,1.2998) -- (9,1.2998);  
\draw [domain=1:7, samples=200] plot (\x,{0.3*sin(4*pi*sqrt(\x)/sqrt(7) r)+1-0.3*sin(4*pi/sqrt(7) r)});  
\draw (-2.7,0) node{$y=0$}; 
\draw (-2.8,1) node{$y=h(x)$}; 
\draw [white,fill=gray!40] (-1,0.9)--(1,0.9)--(1,1.1)--(-1,1.1)--(-1,0.9); 
\draw (0,1.2) node[above]{$b$}; 
\draw [white,fill=gray!40] (3,0.6) circle (0.3);
\draw (3,0.6) node{$f$};
\draw (-1.5,0) node[regular polygon,regular polygon sides=4, fill=black, scale=0.5]{};
\draw (0.5,0) node[regular polygon,regular polygon sides=4, fill=black, scale=0.5]{};
\draw (2.5,0) node[regular polygon,regular polygon sides=4, fill=black, scale=0.5]{};
\draw (4.5,0) node[regular polygon,regular polygon sides=4, fill=black, scale=0.5]{};
\draw (6.5,0) node[regular polygon,regular polygon sides=4, fill=black, scale=0.5]{};
\draw (8.5,0) node[regular polygon,regular polygon sides=4, fill=black, scale=0.5]{};
\draw (-0.5,0) node[regular polygon,regular polygon sides=4, fill=black, scale=0.5]{};
\draw (1.5,0) node[regular polygon,regular polygon sides=4, fill=black, scale=0.5]{};
\draw (3.5,0) node[regular polygon,regular polygon sides=4, fill=black, scale=0.5]{};
\draw (5.5,0) node[regular polygon,regular polygon sides=4, fill=black, scale=0.5]{};
\draw (7.5,0) node[regular polygon,regular polygon sides=4, fill=black, scale=0.5]{};

\end{tikzpicture}
\end{center}
\caption{\label{meas} Parametrization of a slowly variable waveguide of width $h$. A wavefield $u$ is generated by an internal source $f$ and/or a boundary source $b$. Squares represent measurements of $u$ taken on the surface $y=0$.}
\end{figure}
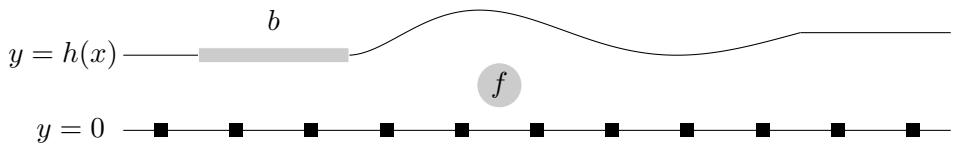

\subsection{Scientific context}

The detection and reconstruction of shape defects in a waveguide are mentioned in different works. In articles \cite{lu1,abra1,abra2}, the authors use a conformal mapping to map the geometry of the perturbed waveguide to that of a regular waveguide. This method is beneficial to understand the propagation of waves in irregular waveguides. Still, it is not easily adaptable to the inverse problem and the reconstruction of defects since the transformation to a regular waveguide is not explicit and proves numerically expensive. Another approach based on the scattering field treatment is developed in \cite{norgren1}. Finally, articles \cite{pagneux2,folguera1} study the forward problem and give leads on how to use one side boundary measurements to reconstruct the width of a slowly varying waveguide.

Our work concerns a different approach, also used in \cite{bao1,bao2}, where we assume the data to be  available for a whole interval of frequencies. This provides additional information that should help to localize and reconstruct the shape of the defect. Moreover, the use of multi-frequency data often provides uniqueness of the reconstruction (see \cite{acosta1}) and a better stability (see \cite{bao3,isakov1,sini1}).

In a previous work \cite{bonnetier1}, we already 
presented a method to reconstruct small width variations using an inverse scattering procedure. However, 
we
avoided all the cut-off frequencies of the waveguide, which are frequencies $k>0$ such that $k=n\pi/h(x^\star_k)$ for a mode $n\in \N$ and a longitudinal position $x^\star_k$. Since experimental works presented in \cite{balogun1,ces1} suggest that these resonant frequencies are helpful to reconstruct width variations, we choose in this article to work only with these frequencies. Using the study of the forward problem already done in \cite{bonnetier2}, we know that if $k$ is a locally resonant frequency, the wavefield $u$ strongly depends on $x^\star_k$. The main idea of our reconstruction method is to use measurements of $u$ to find back $x^\star_k$. Since $h(x^\star_k)=n\pi/k$, it then gives up the information about the width of the waveguide in one point. By taking different locally resonant frequencies $k$, we show that one can obtain a complete approximation of the width $h$ of the waveguide.

\subsection{Outline of the paper}

The key result of this paper is Theorem \ref{th2}, which proves that $u$ is close to a three parameters Airy function, one parameter being $x^\star_k$. As explained in Proposition \ref{leastsq}, it enables us to find the value of $x^\star_k$ for every locally resonant frequency, and to prove that our reconstruction method is $\text{L}^\infty$-stable. The paper is organized as follows. In section 2, we briefly recall results on the modal decomposition and the study of the forward problem. In section 3, we study the inverse problem with measurements taken at the surface of the waveguide and we provide a stability result for the reconstruction of the width of the waveguide. Finally, in section 4, we provide numerical reconstruction of different width defects.

\subsection{Notations}

The varying waveguide is denoted by $\Omega$, its boundary by $\partial\Omega$ and the subscript “top” (resp. “bot”) indicates the upper boundary of the waveguide (resp. lower). We denote $\nu$ the outer normal unit vector. For every $\ell>0$, we set $\Omega_\ell= \{(x,y)\in \Omega\, |\, |x|<\ell\}$. Spaces $\tH^1$ $\tH^2$, $W^{1,1}$, $\tH^{1/2}$ over $\Omega$ or $\R$ are classic Sobolev spaces. The Airy function of the first kind (resp. second kind) is denoted by $\Ai$ (resp. $\Bi$). They are linear independent solutions of the Airy equation $y''-xy=0$ (see \cite{abramowitz1} for more results about Airy functions). See in Figure \ref{airy} the graph of these two functions. The term $\delta_{x=s}$ denotes the Dirac distribution at the point $s\in \R$ and the function $\textbf{1}_{E}$ is the indicator function of the set $E$. Finally, the notation $\{a:b:\ell\}$ designates the uniform discretization of the interval $[a,b]$ with $\ell$ points. 

\begin{figure}[h!]
\begin{center}
\scalebox{.5}{\input{airy}}
\caption{\label{airy} Representation of the Airy functions $\Ai$ and $\Bi$.}
\end{center}
\end{figure}

\section{Brief study of the forward problem}

Before studying the inverse problem associated with the reconstruction of the width in a varying waveguide, we need to study the forward problem in order to find an approximation of the available data. In this section, we briefly recall all the main results on the study of the forward problem. These results and their proofs can be found in \cite{bonnetier1,bourgeois1}.

A useful tool when working in waveguides is the modal decomposition. The following definition provides a modal decomposition in varying waveguides:
\begin{defi}\label{def:modes} We define the sequence of functions $(\ph_n)_{n\in\N}$ by
\begin{equation}\label{phin}
\forall (x,y)\in \Omega,\quad \ph_n(x,y) :=
\left\{\begin{array}{cl}
1/\sqrt{h(x)}\quad  &\text{if } n=0, \\ 
 \frac{\sqrt{2}}{\sqrt{h(x)}}\cos\left(\frac{n\pi y}{h(x)}\right)\quad &\text{if } n\geq 1,
\end{array}\right.
\end{equation}
which for any fixed $x\in \R$ defines an orthonormal basis of $\tL^2(0,h(x))$.\end{defi}

Hence, a solution $u \in \tH^2_{\loc}\big(\Omega\big)$ of \eqref{eqdebut} admits a unique modal decomposition
\begin{equation}\label{decmode}
 u(x,y)=\sum_{n\in \N} u_{n}(x) \ph_n(x,y)\quad\text{where}\quad  u_{n}(x):=\int_0^{h(x)} u(x,y)\ph_n(x,y)\dd y.
\end{equation}
Note that $ u_{n}$ does not satisfy in general any nice equation. However, when $h$ is constant (outside of $\supp(h')$), it satisfies an equation of the form $u_{n}''+k_n^2 u_{n}=- g_n$ where $k_n^2=k^2-n^2\pi^2/h^2$ is the wavenumber. When $h$ is variable, the decomposition \eqref{decmode} motivates the following definition:

\begin{defi}\label{def:wavenumber}  The local wavenumber function of the mode $n\in\N$ is the complex function $k_n:\R\to \C$ defined by
\begin{equation}
k_n^2(x):=k^2-\frac{n^2\pi^2}{h(x)^2},
\end{equation}
with $\text{Re}(k_n), \text{Im}(k_n)\geq 0$. 
\end{defi}

In this work, as $h(x)$ is non constant, $k_n(x)$ may vanish for some $x\in\R$ and change from a positive real number to a purely imaginary one. We then distinguish three different situations: 

\begin{defi}  A mode $n\in\N$ falls in one of these three situations:
\begin{enumerate}
\item If $n<kh(x)/\pi$ for all $x\in\R$ then $k_n(x)\in(0,+\infty)$ for all $x\in\R$ and the mode $n$ is called propagative. 
\item If $n>kh(x)/\pi$ for all $x\in\R$ then $k_n(x)\in i(0,+\infty)$ for all $x\in\R$ and the mode $n$ is called evanescent. 
\item If there exists $x^\star_k\in \R$ such that $n=kh(x^\star_k)/\pi$ the mode $n$ is called locally resonant. Such points $x^\star_k$ are called resonant points, and there are simple if $h'(x^\star_k)\neq 0$, and multiple otherwise. 
\end{enumerate}
A frequency $k>0$ for which there exists at least a locally resonant mode is called a locally resonant frequency.  
\end{defi}

Using the wavenumber function, one can adapt the classic Sommerfeld (or outgoing) condition, defined in \cite{bonnetier1} for regular waveguides, to general varying waveguides $\Omega$. This condition is used to guarantee uniqueness for the source problem given in equation \eqref{eqdebut}.

\begin{defi}\label{def:outgoing}  A wavefield $ u_k \in \tH^2_{\loc}\big( \Omega\big)$ is said to be outgoing if it satisfies 
\begin{equation} \label{sommer}\left|  u_{n}'(x)\frac{x}{|x|}-ik_n(x) u_{n}(x) \right| \underset{|x|\rightarrow +\infty}{\longrightarrow} 0 \qquad \forall n\in \N,
\end{equation}
where $u_{n}$ is given in \eqref{decmode}. 
\end{defi}

In all this work, we make the following assumptions:

\begin{asum}\label{def:slow} We assume that $h\in \mathcal{C}^2(\R)\cap W^{2,\infty}(\R)$ with $h'$ compactly supported and we then define two constants $0<h_{\min}\leq h_{\max}<+\infty$ such that 
\begin{equation}\nonumber
\forall x\in \R \quad h_{\min} \leq h(x) \leq h_{\max},
\end{equation}
and $h(x)=h_{\min}$ or $h(x)=h_{\max}$ if $x\notin \supp(h')$. For such a function, we define a parameter $\eta>0$ that satisfies 
\begin{equation}\nonumber
\Vert h'\Vert_{\text{L}^\infty(\R)} <\eta\quad\text{ and } \quad \Vert h''\Vert_{\text{L}^\infty(\R)}<\eta^2.
\end{equation}
\end{asum}
Such a waveguide is represented in Figure \ref{wg}.

\begin{figure}[h]
\begin{center}
\begin{tikzpicture} 
\draw (-1.5,0) -- (9,0);
\draw (-1.5,1) -- (0,1); 
\draw (7,1.4) -- (9,1.4); 
\draw (4.5,0.6) node{$\Omega$}; 
\draw [domain=0:7, samples=100] plot (\x,{1+12/7/7/7/7/7*\x*\x*\x*(\x*\x/5-7*\x/2+49/3+0.02)});  
\draw (-1.5,1) node[left]{$h_{\min}$};
\draw (9,1.4) node[right]{$h_{\max}$};
\draw [dashed] (0,-0.2)--(0,1.6); 
\draw [dashed] (7,-0.2)--(7,1.6);
\draw [<->] (0.1,-0.1)--(6.9,-0.1); 
\draw (3.5,-0.4) node{supp($h'$)}; 
\end{tikzpicture}
\end{center}
\caption{\label{wg} Representation of an increasing slowly and compactly varying waveguide.}
\end{figure}
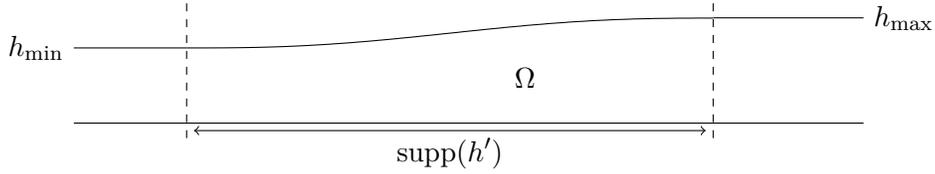

The forward source problem is defined for every frequency $k>0$ by 
\begin{equation}\label{eqmatlab} (\mathcal{H}_k) : \quad 
\left\{\begin{array}{cl} \Delta u +k^2 u= -f  & \text{ in } \Omega, \\
\partial_\nu u =b_\top &\text{ in } \partial\Omega_\text{top},\\
\partial_\nu u =b_\bot &\text{ in } \partial\Omega_\text{bot}, \\
u\text{ is outgoing,} \end{array}\right. 
\end{equation}
As explained in \cite{bourgeois1}, this problem is not well-posed  when $\{x \,|\, k_n(x)=0\}$ is a non-trivial interval of $\R$. This especially happens when $k=n\pi/h_{\min}$ or $k=n\pi/h_{\max}$. We then avoid these two forbidden situations and we set
\begin{equation}\label{delta}
\delta(k):=\min_{n\in \N}\left(\sqrt{\left|k^2-\frac{n^2\pi^2}{{h_{\min}}^2}\right|},\sqrt{\left|k^2-\frac{n^2\pi^2}{{h_{\max}}^2}\right|} \right).
\end{equation}
From now on, we define $(f_n)_{n\in \N}$ the modal decomposition of $f$, and
\begin{equation}\label{gn}
g_n(x):=\frac{f_n(x)}{\sqrt{h(x)}}+\varphi_n(1)b_\top(x)\frac{\sqrt{1+(h'(x))^2}}{\sqrt{h}}+\varphi_n(0)b_\bot(x) \frac{1}{\sqrt{h}}. 
\end{equation}
Using the work done in \cite{bonnetier2}, we are able to provide an approximation of the solution of \eqref{eqmatlab}. If $h$ is increasing, we can state the following result using Theorem 1 and Remark 3 in \cite{bonnetier2}. 

\begin{theorem}\label{th1}
Let $h$ be an increasing function defining a varying waveguide $\Omega$ that satisfies Assumption \ref{def:slow} with a variation parameter $\eta>0$. Consider sources $f\in \tL^\infty_c(\Omega)$, $b:=(b_\bot,b_\top)\in ({\text{H}}^{1/2}(\R))^2\cap (\text{L}^\infty_c(\R)^2)$. Assume that there is a unique locally resonant mode $N\in \N$, associated with a simple resonant point $x^\star_k\in\R$. Let $I\subset \R$ be an interval of length $R>0$, and $\Omega_I:=\{(x,y)\in \Omega \, |\, x\in I\}$. 

There exists $\eta_0>0$ depending only on $h_{\min}$, $h_{\max}$, $\delta(k)$ and $R$ such that if $\eta\leq \eta_0$, then the problem $(\cH_k)$ admits a unique solution $u\in \text{H}^2_{\text{loc}}\big(\Omega)$. Moreover, this solution is approached by $u^{\text{app}}$ defined for almost every $(x,y)\in \Omega$ by 
\begin{equation}\label{greentot}
u^{\text{app}}(x,y):=\sum_{n\in \N} \left(\int_\R G_{n}^{\text{app}}(x,s)g_n(s)\dd s\right)\ph_n\left(y\right),
\end{equation}
where $(f_n)_{n\in \N}$ is the modal decomposition of $f$, $\ph_n$ is defined in \eqref{phin} and $G_{n}^{\text{app}}$ is given by 
\begin{equation}\label{greenfunction}
G_{n}^{\text{app}}(x,s):=
\left\{\begin{aligned}
&\frac{i}{2\sqrt{k_n(s)k_n(x)}}\exp\left(i\left|\int_s^xk_n\right|\right), & \quad\text{ if } n<N,\\
&\frac{1}{2\sqrt{|k_n|(s)|k_n|(x)}}\exp\left(-\left|\int_s^x|k_n|\right|\right), & \quad\text{if } n>N,\\
&\left\{\begin{aligned}
\frac{\pi(\xi(s)\xi(x))^{1/4}}{\sqrt{k_n(s)k_n(x)}}\big(i\Ai+\Bi\big)\circ\xi(s)\Ai\circ\xi(x)& \quad\text{ if } x<s, \\
\frac{\pi(\xi(s)\xi(x))^{1/4}}{\sqrt{k_n(s)k_n(x)}}\big(i\Ai+\Bi\big)\circ\xi(x)\Ai\circ\xi(s)& \quad\text{ if } x>s, \\
\end{aligned}\right.
&\quad\text{ if } n=N.
\end{aligned}\right.\end{equation}
Function $k_n$ is the wavenumber function defined in Definition \ref{def:wavenumber} and the function $\xi$ is given by 
\begin{equation}\label{eq:xi}
\xi(x):=
\left\{\begin{aligned}
\left(-\frac{3}{2}i\int_x^{x^\star_k}k_N(t)\dd t\right)^{2/3} & \text{ if } x<x^\star_k, \\ -\left(\frac{3}{2}\int_{x^\star_k}^x k_N(t) \dd t \right)^{2/3} & \text{ if } x>x^\star_k.
\end{aligned}\right.\end{equation}
Precisely, there exist a constant $C_1>0$ depending only on $h_{\min}$, $h_{\max}$ and $N$ such that
\begin{equation}
\Vert u -u^{\text{app}}\Vert_{\text{H}^1(\Omega_I)}\leq C_1\eta R^2\delta(k)^{-8}\left(\Vert f\Vert_{\text{L}^2(\Omega)}+\Vert b\Vert_{\left({\text{H}}^{1/2}(\R)\right)^2}\right).
\end{equation}

\end{theorem}

This result provides an approximation of the measurements of the wavefield for every frequency, and a control of the approximation error. We represent in Figure \ref{direct_multi_freq} the wavefield for different frequencies. We also point out in Figure \ref{direct_multi_freq2} that the source should be located in an area where $h>h(x^\star)$ in order to generated a significant locally resonant mode.

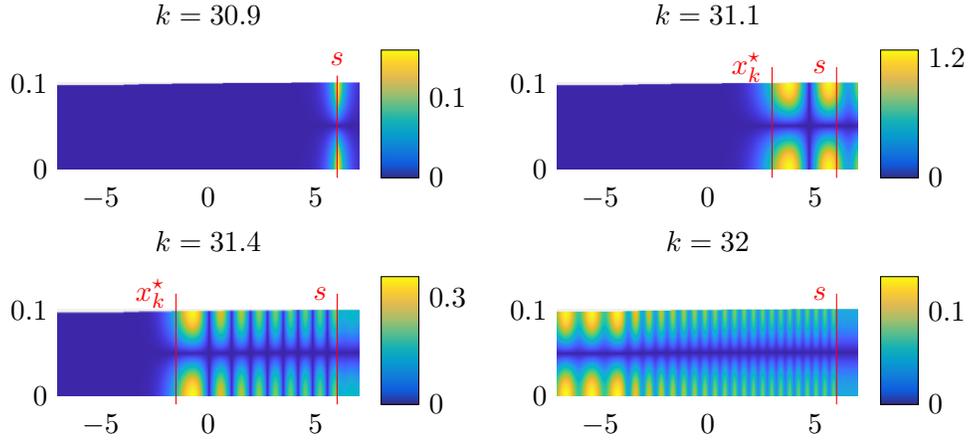
\begin{figure}\begin{center}
\begin{tikzpicture}
\begin{axis}[width=4cm, height=1.7cm, axis on top, scale only axis, xmin=-7, xmax=7, ymin=-0.01, ymax=0.14, colorbar,point meta min=0,point meta max=0.1604, title={$k=30.9$}, colorbar style={ytick distance=0.1}, ytick distance=0.1,axis line style={draw=none},tick style={draw=none}]
\addplot graphics [xmin=-7,xmax=7,ymin=0,ymax=0.1013]{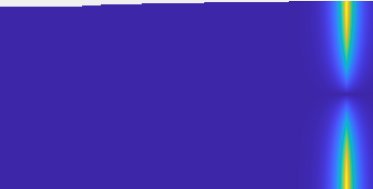};
\draw [red] (6,-0.01)--(6,0.11) node[above]{$s$};
\end{axis} 
\end{tikzpicture}\hspace{2mm}
\begin{tikzpicture}
\begin{axis}[width=4cm, height=1.7cm, axis on top, scale only axis, xmin=-7, xmax=7, ymin=-0.01, ymax=0.14, colorbar,point meta min=0,point meta max=1.2831, title={$k=31.1$}, colorbar style={ytick distance=1.2}, ytick distance=0.1,axis line style={draw=none},tick style={draw=none}]
\addplot graphics [xmin=-7,xmax=7,ymin=0,ymax=0.1013]{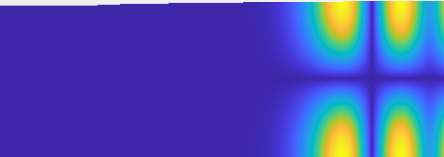};
\draw [red] (6,-0.01)--(6,0.12) node[left]{$s$};
\draw [red] (3,-0.01)--(3,0.12) node[left]{$x^\star_k$};
\end{axis} 
\end{tikzpicture}

\begin{tikzpicture}
\begin{axis}[width=4cm, height=1.7cm, axis on top, scale only axis, xmin=-7, xmax=7, ymin=-0.01, ymax=0.14, colorbar,point meta min=0,point meta max=0.3587, title={$k=31.4$}, colorbar style={ytick distance=0.3}, ytick distance=0.1,axis line style={draw=none},tick style={draw=none}]
\addplot graphics [xmin=-7,xmax=7,ymin=0,ymax=0.1013]{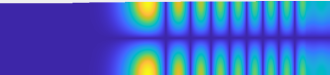};
\draw [red] (6,-0.01)--(6,0.12) node[left]{$s$};
\draw [red] (-1.5,-0.01)--(-1.5,0.12) node[left]{$x^\star_k$};
\end{axis} 
\end{tikzpicture}\hspace{2mm}
\begin{tikzpicture}
\begin{axis}[width=4cm, height=1.7cm, axis on top, scale only axis, xmin=-7, xmax=7, ymin=-0.01, ymax=0.14, colorbar,point meta min=0,point meta max=0.1383, title={$k=32$}, colorbar style={ytick distance=0.1}, ytick distance=0.1,axis line style={draw=none},tick style={draw=none}]
\addplot graphics [xmin=-7,xmax=7,ymin=0,ymax=0.1013]{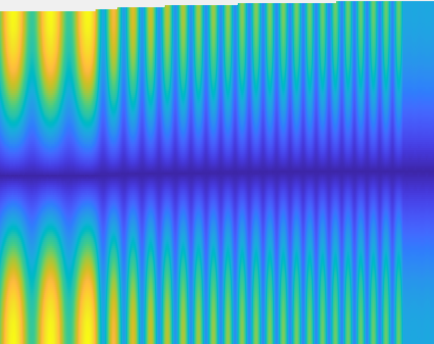};
\draw [red] (6,-0.01)--(6,0.12) node[left]{$s$};
\end{axis} 
\end{tikzpicture}
\end{center}
\caption{\label{direct_multi_freq} Representation of the wavefield amplitude $|u|$ solution of \eqref{eqmatlab} with a monochromatic source $f(x,y):=\delta_{x=s} \varphi_1(y)$ for different frequencies $k$. Data are generated using the finite element method described in Section \ref{num}. When $k=30.9$, the mode $n=1$ is evanescent, and the wavefield decreases very fast around the source. When $k=31.1$ and $k=31.4$, the mode $n=1$ is locally resonant and the wavefield propagates in the waveguide as an Airy function until it reaches the point $x^\star_k$. When $k=32$, the mode $n=2$ is propagative and the wavefield propagates in all the waveguide. In all these representations, the width $h$ is defined by the function \eqref{h4}. }
\end{figure}

\begin{figure}\begin{center}
\begin{tikzpicture}
\begin{axis}[width=4cm, height=1.5cm, axis on top, scale only axis, xmin=-6, xmax=6, ymin=-0.01, ymax=0.14, colorbar,point meta min=0,point meta max=0.3188, title={$k=31$}, colorbar style={ytick distance=0.3}, ytick distance=0.1,axis line style={draw=none},tick style={draw=none}]
\addplot graphics [xmin=-7,xmax=7,ymin=0,ymax=0.1025]{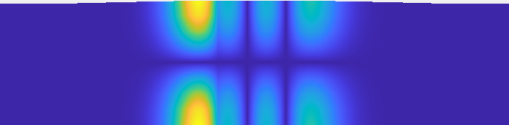};
\draw [red] (-1,-0.01)--(-1,0.125) node[right]{$s$};
\draw [red] (-2,-0.01)--(-2,0.125) node[left]{$x^\star_k$};
\draw [red] (2,-0.01)--(2,0.125) node[right]{$x^\star_k$};
\end{axis} 
\end{tikzpicture}\hspace{2mm}
\begin{tikzpicture}
\begin{axis}[width=4cm, height=1.5cm, axis on top, scale only axis, xmin=-6, xmax=6, ymin=-0.01, ymax=0.14, colorbar,point meta min=0,point meta max=0.1002, title={$k=31$}, colorbar style={ytick distance=0.1}, ytick distance=0.1,axis line style={draw=none},tick style={draw=none}]
\addplot graphics [xmin=-7,xmax=7,ymin=0,ymax=0.1025]{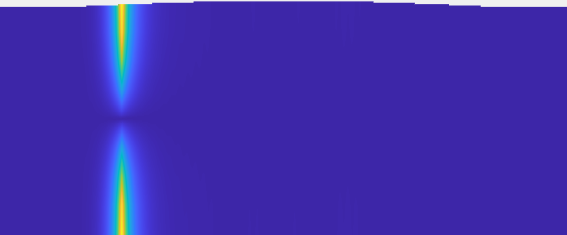};
\draw [red] (-4,-0.01)--(-4,0.12) node[left]{$s$};
\draw [red] (-2,-0.01)--(-2,0.125) node[right]{$x^\star_k$};
\draw [red] (2,-0.01)--(2,0.125) node[right]{$x^\star_k$};
\end{axis} 
\end{tikzpicture}

\begin{tikzpicture}
\begin{axis}[width=4cm, height=1.5cm, axis on top, scale only axis, xmin=-7, xmax=7, ymin=-0.01, ymax=0.14, colorbar,point meta min=0,point meta max= 0.2803, title={$k=32$}, colorbar style={ytick distance=0.2}, ytick distance=0.1,axis line style={draw=none},tick style={draw=none}]
\addplot graphics [xmin=-7,xmax=7,ymin=0,ymax=0.1]{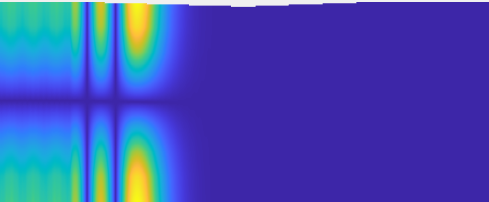};
\draw [red] (-5,-0.01)--(-5,0.12) node[left]{$s$};
\draw [red] (-2.5,-0.01)--(-2.5,0.12) node[right]{$x^\star_k$};
\draw [red] (2,-0.01)--(2,0.12) node[right]{$x^\star_k$};
\end{axis} 
\end{tikzpicture}\hspace{2mm}
\begin{tikzpicture}
\begin{axis}[width=4cm, height=1.5cm, axis on top, scale only axis, xmin=-7, xmax=7, ymin=-0.01, ymax=0.14, colorbar,point meta min=0,point meta max=0.3037, title={$k=32$}, colorbar style={ytick distance=0.3}, ytick distance=0.1,axis line style={draw=none},tick style={draw=none}]
\addplot graphics [xmin=-7,xmax=7,ymin=0,ymax=0.1]{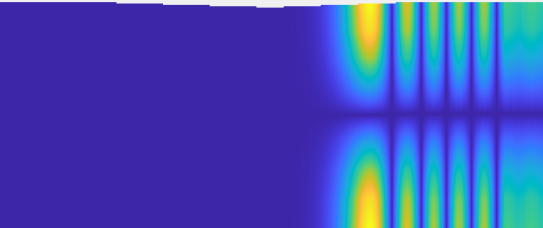};
\draw [red] (6,-0.01)--(6,0.12) node[left]{$s$};
\draw [red] (2,-0.01)--(2,0.12) node[left]{$x^\star_k$};
\draw [red] (-2.5,-0.01)--(-2.5,0.12) node[left]{$x^\star_k$};
\end{axis} 
\end{tikzpicture}
\end{center}
\caption{\label{direct_multi_freq2} Representation of the wavefield amplitude $|u|$ solution of \eqref{eqmatlab} with a monochromatic source $f(x,y):=\delta_{x=s} \varphi_1(y)$ for different positions $s$.  On top, $h$ is a dilation defined in \eqref{h5} and on bottom, $h$ is a shrinkage defined in \eqref{h6}. Depending on the location of the source, we observe different behaviors: while top left, bottom left and bottom right show locally resonant modes, the picture on top right show an evanescent mode. In order to generate a significant locally resonant mode, the source $s$ should be placed at a width where $h(s)>h(x^\star_k)$ and the mode only propagates until it reaches the point $x^\star_k$.}
\end{figure}
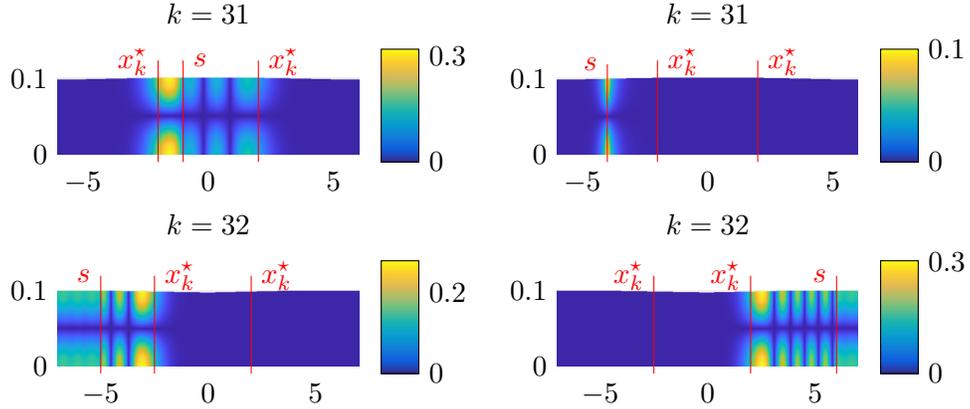

We notice that at locally resonant frequencies, the wavefield strongly depends on the position of $x^\star_k$, which justify the idea of using it to develop an inverse method to reconstruct the width $h$. An illustration is provided in Figure \ref{direct} with a representation of the wavefield $u$ and the one side boundary measurement $u(x,0)$ when $k$ is a locally resonant frequency.

\begin{figure}[h]
\begin{center}
\input{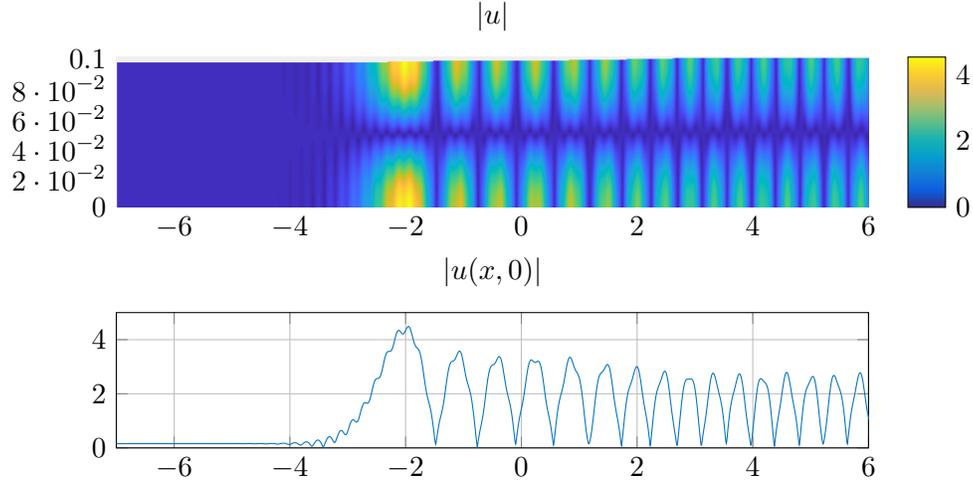}
\end{center}
\caption{\label{direct} Illustration of the amplitude of the wavefield measurements in a varying waveguide. The wavefield $u$ is solution of \eqref{eqmatlab} where the profile $h$ is given in \eqref{h4} and $k=31.5$ is a locally resonant frequency. Data are generated using the finite element method described in Section \ref{num}. Top: amplitude of $|u|$ in the whole waveguide $\Omega$ (non scaled). Bottom: amplitude of the measurements of $|u|$ on the surface $y=0$.}
\end{figure}

\section{Shape Inversion using a monochromatic source \label{section2} }

In this section, we describe the method to recover the width $h$ from one side boundary measurements of the wavefield at locally resonant frequencies. We focus here on the reconstruction of the shape of $h$ on $\supp(h')$, assuming the \emph{a priori} knowledge of the constants $h_{\min}$ and $h_{\max}$ and of an interval containing $\supp(h')$. We detail in Appendix~A how these constants can be estimated. To recover the shape function $h$ on $\supp(h')$, we start by studying the simpler case of a source term $f$ generating only a single locally resonant mode in the waveguide. Hence we assume that $f$ takes the form
\begin{equation}\label{fN}
\forall (x,y)\in \Omega \quad f(x,y)=f_N(x) \varphi_N(y), \qquad f_N\in \text{L}^2(\R), 
\end{equation}
where $f_N$ is compactly supported. We also assume the absence of boundary source term, meaning that $b=0$. This simplified situation is useful to understand the method of reconstruction. It will be generalized to any kind on internal and boundary sources in section \ref{section_general_source}.

In the simplest case of a single internal source \eqref{fN}, we know from the study of the forward problem in Theorem \ref{th1} that the measured data without noise $d^\text{ex}$ satisfies 
\begin{equation} \label{dex}
d^\text{ex}(x):=u(x,0)\approx u_{N}^{\app}(x) \varphi_N(0) \quad \text{where} \quad u_{N}^\app(x)=\int_\R G_{N}^\app(x,s)f_N(s)\dd s.
\end{equation}
Our reconstruction method is based on the recovery of the resonant point $x^\star_k$ for every locally resonant frequencies $k$. It can be seen in Theorem \ref{th1} that the approached Green function $G_{N}^\app $ depends on $x^\star_k$ through the function $\xi$. However, this dependence is intricate and hardly usable to find a direct link between $d^\text{ex}$ and $x^\star_k$. Thus, we need to find a simpler approximation of the measurements. In a first part, we provide an approximation of the measured data $d^\text{ex}$ using a three parameters model function
\begin{equation}\label{dapp}
d^\text{app}_{z,\alpha,x_k^\star} : x\mapsto z\Ai(\alpha(x_k^\star-x)),
\end{equation}
where $z\in \C^*$ is a complex amplitude, $\alpha>0$ is a scaling parameter and $x_k^\star$ is the resonant point playing the role of a longitudinal shift. In a second part, we first control the approximation error between this function and the exact measurements and then, we develop a stable way to reconstruct $x_k^\star$ from this approximated data.

\subsection{Wavefield approximation and measurements approximation}

In order to find a reconstruction of $x^\star_k$, we need to find an exploitable link between $d^\ex:=u(x,0)$ and $x^\star_k$. To do so, we make a Taylor expansion of $G_{N}^{\app}$ around the point $x^\star_k$. For every frequency $k>0$ and $R>0$, we denote
\begin{equation}
\Omega_R(x^\star_k):=\{(x,y)\in \Omega \ |\ |x-x^\star_k| < R\}, \qquad \Gamma_R(x^\star_k):=(x^\star_k-R,x^\star_k+R),
\end{equation}
and we consider the Taylor expansion on the interval $\Gamma_R$. Moreover, we assume that the source is located at the right of the interval $\Gamma_R$, and we thus we define 
\begin{equation}
\Omega_R^+(x^\star_k):=\{(x,y)\in \Omega \, |\, x-x^\star_k> R\},\qquad \Gamma_R^+(x^\star_k)=(x^\star_k+R,+\infty).
\end{equation}
Both these sets are represented in Figure \ref{lambda_R}.

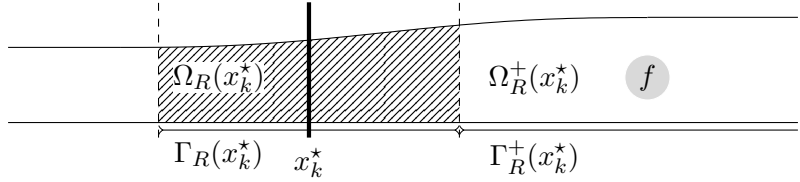
\begin{figure}[h]
\begin{center}
\begin{tikzpicture} 
\draw (-1.5,0) -- (9,0);
\draw (-1.5,1) -- (0,1); 
\draw (7,1.4) -- (9,1.4); 
\draw [domain=0:7, samples=100] plot (\x,{1+12/7/7/7/7/7*\x*\x*\x*(\x*\x/5-7*\x/2+49/3+0.02)});  
\draw [ultra thick] (2.5,-0.2)--(2.5,1.6); 
\draw (2.5,-0.2) node[below]{$x^\star_k$};
\draw [>-<] (0.5,-0.1)--(4.5,-0.1); 
\draw (1.3,-0.1) node[below]{$\Gamma_R(x^\star_k)$}; 
\draw [dashed] (0.5,-0.2)--(0.5,1.6); 
\draw [dashed] (4.5,-0.2)--(4.5,1.6);
\fill[color=gray!20, pattern=north east lines] (0.5,0)--(0.5,1.0013) --  plot [domain=0.5:4.5] (\x,{1+12/7/7/7/7/7*\x*\x*\x*(\x*\x/5-7*\x/2+49/3+0.02)}) --(4.5,1.3028) -- (4.5,0)--(0.5,0);
\fill[white] (1.9,0.4)--(0.7,0.4)--(0.7,0.8)--(1.9,0.8)--(1.9,0.4); 
\draw (1.3,0.6) node{$\Omega_R(x^\star_k)$};
\draw [white,fill=gray!30] (7,0.6) circle (0.3);
\draw (7,0.6) node{$f$};
\draw [>-] (4.5,-0.1)--(9,-0.1); 
\draw (5.5,-0.1) node[below]{$\Gamma_R^+(x^\star_k)$};
\draw (5.5,0.6) node{$\Omega_R^+(x^\star_k)$}; 
\end{tikzpicture}
\end{center}
\caption{\label{lambda_R} Representation of $\Gamma_R(x^\star_k)$, $\Omega_R(x^\star_k)$, $\Gamma_R^+(x^\star_k)$ and $\Omega_R^+(x^\star_k)$. The source $f$ is assumed to be compactly supported in $\Gamma_R^+(x^\star_k)$.}
\end{figure}

The following Proposition shows that $u_{N}^\app$ can be approached by a three parameters function of type $d^\text{app}_{z,\alpha,x_k^\star}$ if the shape function $h$ is steep enough at $x_k^\star$.

\begin{prop} \label{taylor} Assume that $\Omega$ satisfies assumptions \ref{def:slow}, let $R>0$ and $k>0$ be a locally resonant frequency associated to the mode $N\in \N$ and the locally resonant point $x_k^\star$ that satisfies $h'(x^\star_k)\geq \theta\eta$ for some $\theta>0$. There exist $z\in \C^*$, $\alpha>0$ and a constant $C_2>0$ depending only on $h_{\min}$, $h_{\max}$, $N$ and $\theta$ such that
\begin{equation}\label{taylorapp}
\Vert u^\text{app}_{N}-d^\text{app}_{z,\alpha,x_k^\star}\Vert_{\text{L}^2(\Gamma_{R}(x^\star_k))}\leq C_2\left(R^{3/2}\eta^{5/6}+R^{5/2}\eta^{7/6}\right),
\end{equation}
where $d^\text{app}_{z,\alpha,x_k^\star}$ is defined in \eqref{dapp}. 
\end{prop}

\begin{proof}
Using the information about the support of the source term, we know that 
\begin{equation}\nonumber
u_{N}^{\app}(x)=\int_{x^\star_k+R}^{+\infty}G_{N}^{\app}(x,s)f_N(s)\dd s.
\end{equation}
Using the definition of $G_{N}^\app$ given in \eqref{greenfunction}, there exists a function $q_k$ such that for every $x\in \Gamma_R(x^\star_k)$,
\begin{equation}\nonumber
G_{N}^{\app}(x,s)=q_k(s)\frac{(-\xi(x))^{1/4}}{\sqrt{k_n(x)}}\Ai(\xi(x)).
\end{equation}
It follows that 
\begin{equation}\nonumber 
u_{N}^{\app}(x)=\frac{(-\xi(x))^{1/4}}{\sqrt{k_n(x)}}\Ai(\xi(x))\int_{x^\star_k+R}^{+\infty}q_k(s)f_N(s)\dd s.
\end{equation}
In the following, we denote $\cO(\cdot)$ bounds depending only on $h_{\max}$, $h_{\min}$ and $N$. We see that 
\begin{equation}\nonumber
k_N(x)^2=\frac{2N^2\pi^2h'(x^\star_k)}{h(x^\star_k)^3}(x-x^\star_k) +\cO(\eta^2(x-x^\star_k)^2).
\end{equation}
From now on, we assume that $x>x^\star_k$, which leads to 
\begin{equation}\label{compkn}
k_N(x)=\sqrt{\frac{2N^2\pi^2h'(x^\star_k)}{h(x^\star_k)^3}}(x-x^\star_k)^{1/2} +\cO\left(\frac{\eta^2(x-x^\star_k)^{3/2}}{\sqrt{h'(x^\star_k)}}\right),
\end{equation}
and so
\begin{equation}\nonumber
\xi(x)=\left(\frac{2N^2\pi^2h'(x^\star_k)}{h(x^\star_k)^3}\right)^{1/3}(x^\star_k-x)+\cO\left(\frac{\eta^2(x-x^\star_k)^{2}}{h'(x^\star_k)^{2/3}}\right).
\end{equation}
Then, 
\begin{equation}\nonumber
\Ai(\xi(x))=\Ai\left(\left(\frac{2N^2\pi^2h'(x^\star_k)}{h(x^\star_k)^3}\right)^{1/3}(x^\star_k-x)\right)+\cO\left(\frac{\eta^2(x-x^\star_k)^{2}}{h'(x^\star_k)^{2/3}}\right),
\end{equation}
and 
\begin{equation}\nonumber
\frac{(-\xi(x))^{1/4}}{\sqrt{k_N(x)}}=\left(\frac{2N^2\pi^2h'(x^\star_k)}{h(x^\star_k)^3}\right)^{-1/6}+\cO\left(\frac{\eta^2(x-x^\star_k)}{h'(x^\star_k)^{7/6}}\right).
\end{equation}
We set
\begin{equation}\label{exppk}
z=\left(\frac{2N^2\pi^2h'(x^\star_k)}{h(x^\star_k)^3}\right)^{-1/6}\int_{x^\star_k+R}^{+\infty}q_k(s)f_N(s)\dd s, \qquad \alpha=\left(\frac{2N^2\pi^2h'(x^\star_k)}{h(x^\star_k)^3}\right)^{1/3}, 
\end{equation}
and it follows that 
\begin{equation}\nonumber
u_{N}^{\app}(x)=z\Ai(\alpha(x^\star_k-x))+\cO\left(\frac{\eta^2(x-x^\star_k)}{h'(x^\star_k)^{7/6}}\right)+\cO\left(\frac{\eta^2(x-x^\star_k)^2}{h'(x^\star_k)^{5/6}}\right).
\end{equation}
The exact same study can be done in the case $x<x^\star_k$. Since $|x-x^\star_k|\leq R$, we conclude that 
\begin{equation}\nonumber
\Vert u_{N}^{\app}-d^\text{app}_{z,\alpha,x_k^\star}\Vert_{\text{L}^2(\Gamma_{R}(x^\star_k))}=\eta^2\cO\left(R^{3/2}(h'(x^\star_k))^{-7/6}+R^{5/2}(h'(x^\star_k))^{-5/6}\right).
\end{equation}
To conclude, we use the fact that $h'(x^\star_k)\geq \theta\eta$. 
\end{proof}

\begin{coro}\label{corobis}
Assume that $\Omega$ satisfies assumptions \ref{def:slow}, let $R>0$ and $k>0$ be a locally resonant frequency associated to the mode $N\in \N$ and the locally resonant point $x_k^\star$ that satisfies $h'(x^\star_k)\geq \theta\eta$ for some $\theta>0$. There exist $\eta>0$ such that if $\eta\leq \eta_0$ then there exist $z\in \C^*$, $\alpha>0$ and a constant $C_3>0$ depending only on $h_{\min}$, $h_{\max}$, $N$ and $\theta$ such that
\begin{equation}
\Vert d^\text{ex} - d^\text{app}_{z,\alpha,x_k^\star} \Vert_{\text{L}^2(\Gamma_R(x^\star_k))} \leq C_3\left( \delta(k)^{-8}R^2\eta+R^{3/2}\eta^{5/6}+R^{5/2}\eta^{7/6}\right).
\end{equation}
\end{coro}

\begin{proof}
We apply Proposition \ref{taylor}, Theorem \ref{th1} and trace results which prove that there exists a constant $\gamma>0$ depending only on $R$, $h_{\min}$ and $h_{\max}$ such that 
\begin{equation}\nonumber 
\Vert u(\cdot,0)-u^{\app}(\cdot,0)\Vert_{\text{L}^2(\Gamma_{R}(x^\star_k))}\leq \Vert u-u^{\app}\Vert_{\text{H}^{1/2}(\Gamma_{R}(x^\star_k))}\leq \gamma\Vert u-u^{\app}\Vert_{\text{H}^{1}(\Omega_{R}(x^\star_k))}.
\end{equation}
\end{proof}

This result indicates a strategy to determine $x_k^\star$ from the exact data $d^\text{ex}$. Assuming that $\eta$ is small enough to see $d^\text{app}_{z,\alpha,x_k^\star}$ as a good approximation of $d^\text{ex}$, one may fit the three parameters $(\alpha,z,x_k^\star)$ that minimize the misfit $d^\text{ex} - d^\text{app}_{z,\alpha,x_k^\star}$ for some norm. We also see from this result that the error in the Taylor expansion strongly depends on the values of $\theta$, $R$ and $\delta(k)$. Since we plan on using different locally resonant frequencies, we need to get a uniform control over $k$ on the Taylor expansion. Moreover, we need to control the length $R$ of the measurement interval. If $R$ is too large, the quality of the Taylor expansion diminishes and if $R$ is too small, the data on the interval $\Gamma_{R}(x^\star_k)$ may not contain enough information to fit the three parameters $(z,\alpha,x_k^\star)$ with stability. 

To prevent this, $R$ needs to be scaled depending on the parameter $\alpha$. Indeed, the change of variable $x\mapsto\alpha(x^\star_k-x)$ must cover a large enough interval to perform the desired fitting. Using the expression of $\alpha$ given in \eqref{exppk}, we say that $\alpha(x^\star_k-x)$ covers an interval of fixed radius $\sigma>0$ if 
\begin{equation}\nonumber
R\geq \frac\sigma\alpha\geq\frac{\sigma h_{\max}}{2\beta^{1/3}N^2\pi^2}\eta^{-1/3}. 
\end{equation}
This means that $R$ needs to be scaled as $\eta^{-1/3}$ and we assume now that
\begin{equation}
R:=r\eta^{-1/3}. 
\end{equation}
where $r>0$ is a constant. We can now give a uniform approximation result between $d^\text{ex}$ and $d^\text{app}_{z,\alpha,x_k^\star}$ on the interval $\Gamma_{R(x^\star_k)}$.

\begin{theorem}\label{th2} Assume the same hypotheses than in Theorem \ref{th1}, and fix $N\in \N^*$ and $\theta>0$. There exists $\eta_0>0$ such for any $\eta <\eta_0$ and for any $k>0$ associated to the mode $N$ and the locally resonant point $x_k^\star$ satisfying $h'(x^\star_k)\geq \theta\eta$, there exist $\alpha>0$ and $z\in \C$ such that the following approximation holds on the interval $\Gamma_{R}(x^\star_k)$:
\begin{equation}\nonumber 
d^\text{ex} = d^\text{app}_{z,\alpha,x_k^\star} + \cO(\eta^{1/3}).
\end{equation}
More precisely, there exists a constant $C_4>0$ depending only on $h_{\min}$, $h_{\max}$, $\theta$ and $r$ such that
\begin{equation}
\norm{d^\text{ex} - d^\text{app}_{z,\alpha,x_k^\star}}{\text{L}^2(\Gamma_{R}(x^\star_k))}\leq C_4\eta^{1/3}. 
\end{equation}
\end{theorem}

\begin{proof}

To prove this result, we need to provide a uniform lower bound for $\delta(k)$ under the hypothesis $h'(x^\star_k)\geq \theta\eta$. From the definition of $\delta(k)$, we assume without loss of generality that $\delta(k)$ is reached by the mode number $N\in \N$ and that $\delta(k)^2=\left|k^2-N^2\pi^2/h_{\max}^2\right|$. Hence, as $k$ is locally resonant, we know that $k=N\pi/h(x_k^\star)$ and that
\begin{equation}\nonumber
\delta(k)^2=N^2\pi^2\left|\frac{1}{h(x_k^\star)^2}-\frac{1}{h_{\max}^2}\right|\leq \frac{2N^2\pi^2}{   h_{\min}^2h_{\max}}(h_{\max}-h(x_k^\star)). \end{equation}
Let us call $t\geq 0$ such that $h(x_k^\star+t)=h_{\max}$, we necessarily have $h'(x_k^\star+t) = 0$. Using the hypothesis $h''\geq -\eta^2$, we know that $h'(x_k^\star+s)\geq \eta\theta-\eta^2s$ for all $s\in (0,t)$. This implies that $t\geq \theta/\eta$. Then,
\begin{equation}\nonumber
h_{\max}-h(x_k^\star) = \int_0^th'(x_k^\star+s)\dd s\geq \eta\theta t - \eta^2\frac {t^2}2\geq \frac{\theta^2}2.
\end{equation}
Then, 
\begin{equation}\nonumber
\delta(k)\geq \frac{N\pi}{h_{\min}\sqrt{h_{\max}}}\theta. 
\end{equation}
The proof of Theorem \ref{th2} is now straightforward, we simply replace $R$ by $r\eta^{-1/3}$ in Corollary \ref{corobis} and use the lower bound on $\delta(k)$. 
\end{proof}

\begin{rem}
We underline the fact that the constant in this estimation depends on the value of $\theta$, and tends toward infinity if $\theta$ tends to zero. This result is illustrated in Figure \ref{etanu} where we clearly notice that the error between $d^\ex$ and $d^\app_{z,\alpha,x^\star_k}$ deteriorates when $\theta$ become too small. 
\end{rem}

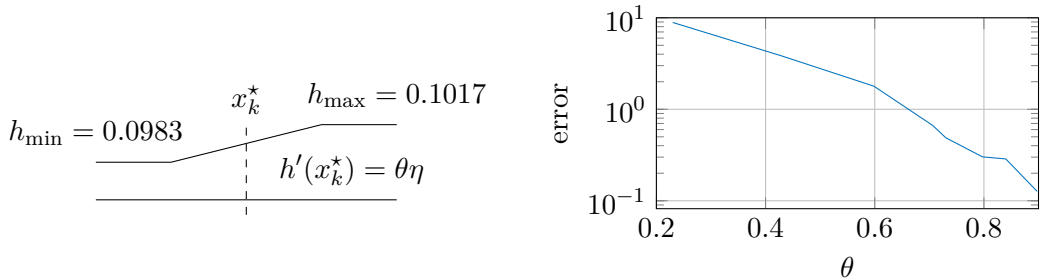
\begin{figure}[h]
\begin{center}
\raisebox{0.5\height}{\begin{tikzpicture}
\draw (-2,0.5)--(-1,0.5)--(1,1)--(2,1); 
\draw (-2,0)--(2,0); 
\draw [dashed] (0,-0.2)--(0,1) node[above]{$x^\star_k$}; 
\draw (0.3,0.4) node[right]{$h'(x^\star_k)=\theta \eta$}; 
\draw (-2,0.6) node[above]{$h_{\min}=0.0983$}; 
\draw (2,1.1) node[above]{$h_{\max}=0.1017$};
\end{tikzpicture}} \hspace{4mm}
%
%
\definecolor{mycolor1}{rgb}{0.00000,0.44700,0.74100}%
\definecolor{mycolor2}{rgb}{0.85000,0.32500,0.09800}%
\begin{tikzpicture}

\begin{semilogyaxis}[%
width=2in,
height=1in,
at={(0in,0in)},
scale only axis,
xmin=0.2,
legend pos=south east,
xmax=0.9,
ymin=0,
xlabel={$\theta$},
ylabel={error},
ymax=10,
grid=major, 
axis background/.style={fill=white},
]
\addplot [color=mycolor1]
  table[row sep=crcr]{%
0.229893184657038	8.88892962598819\\
0.429325049077456	3.83414957190801\\
0.598270761418748	1.79049908594031\\
0.706212630143859	0.665514944667986\\
0.730105393136942	0.488550243946296\\
0.79729674891956	0.3020400865549\\
0.840071041744669	0.286477697737975\\
0.897477962259253	0.126797829758532\\
};
\end{semilogyaxis}

\end{tikzpicture}\end{center}
\caption{\label{etanu} Error of approximation $\Vert d^\ex-d^\app_{z,\alpha,x^\star_k} \Vert_{\text{L}^2(\Gamma_R(x^\star_k))}$ for different values of $\theta$ with a fixed value of $\eta$. The width $h$ is represented on the left of the picture and exact data $d^\ex$ are generated as explained in section \ref{num} with a locally resonant mode $N=1$ and a source $f(x,y)=\delta_{x=6}\varphi_N(y)$. Then, $d^\ex$ is compared with $d^\app_{z,\alpha,x^\star_k}$ where $\alpha$ and $z$ are defined using the expression \eqref{exppk}. Here, $r=0.2$ and $\eta=8.10^{-4}$.}
\end{figure}

\subsection{Stable reconstruction of $x^\star_k$}

We proved so far that the one side boundary measurements are close to the three parameters function

\begin{equation}
d^\text{app}_{z,\alpha,x_k^\star} : x\mapsto z\Ai(\alpha(x_k^\star-x)),
\end{equation}
on the interval $\Gamma_{r\eta^{-1/3}}(x^\star_k)=(x_k^\star-r\eta^{-1/3},x_k^\star+r\eta^{-1/3})$. The question is now to understand if one can get a stable evaluation of $x_k^\star$ from this approached data. We define then the following forward operator 

\begin{equation}\label{eq:F}
F:\left\{\begin{array}{rcl}\C^*\times(0,+\infty)\times\R &\longrightarrow &\cC^0(\Gamma_{R}(\beta/\alpha))\\
(z,\alpha,\beta) &\longmapsto & \left(x\mapsto z\Ai(\beta-\alpha x)\right)\end{array}\right. .
\end{equation}
If the three parameters are uniquely defined, we can deduce an approximation of the position $x_k^\star$ from the solution of $F(z,\alpha,\beta)=d^\text{app}_{z,\alpha,x^\star_k}$ through the identification $x_k^\star=\beta/\alpha$. The following result guarantee the uniqueness of the solution to this problem. 

\begin{prop}  Let $d_0:=F(z_0,\alpha_0,\beta_0)$, there exists $r_0>0$ such that if $r>r_0$ then the problem $F(z,\alpha,\beta)=d_0$ has a unique solution $(z_0,\alpha_0,\beta_0)$. 
\end{prop}

\begin{proof} If $r$ is high enough, the interval $\Gamma_{r\eta^{-1/3}}$ is large enough to contain the maximum amplitude of the function $d_0$ at position $x_{\max}$, and the two first zeros of the function $d_0$ called $x_1$ and $x_2$. Hence $(z_0,\alpha_0,\beta_0)$ are uniquely determined by
\begin{equation}\label{meth11}
z_0=\frac{d_0(x_{\max})}{\norm{\Ai}{\infty}},\qquad \alpha_0=\frac{y_1-y_2}{x_2-x_1}, \qquad \beta_0=\frac{y_1x_2-y_2x_1}{x_2-x_1},
\end{equation}
where $y_1$ and $y_2$ are the two first zeros of the Airy function $\Ai$.
\end{proof}

This result of uniqueness and the corresponding inversion formulas are not directly applicable. Indeed, these formulas are not robust to noise or data uncertainties and we remind the reader that an approximation is already made between $d^\text{ex}$ and $d^\text{app}_{z,\alpha,x^\star_k}$. 

A first strategy would be to use the explicit formula \eqref{meth11} and apply some data regularization with a low-pass filter before doing the inversion (see \cite{mallat1} for more details). This method works perfectly with exact data, and will be called from now the ``direct parameters reconstruction method''.  

However, few issues can be raised with this method. Firstly, it is not sufficiently robust if data are very noisy, as illustrated in Figure \ref{ercomp} where we use this method to reconstruct parameters with a random additive noise of increasing amplitude. Secondly, as mentioned in the previous subsection, we will only make measurements of the wavefield on the interval $\Gamma_R$, and we cannot be sure that neither the first two zeros nor the maximum of $\Ai$ will occur in this interval. Finally, it is more realistic to assume that we only have access to $d^\text{app}_{z,\alpha,x^\star_k}$ on a finite number of receivers positions. For all this reasons, a least squares approach is introduced. On can keep the previous method to find an initial guess of the parameters $\g p:=(z,\alpha,\beta)$. 

We now assume that we have access to the data 
\begin{equation}
d_i:=F(p)(t_i),\quad i=1,\dots,n,
\end{equation}
where $\g p:=(z,\alpha,\beta)$ and $\g t:=(t_i)_{i=1,\dots,n}\in \Gamma_R$ is an increasing subdivision of $\Gamma_R$. We denote by $\g d:=(d_1,\dots,d_n)\in \C^n$ the discrete data on the subdivision $\g t$. We aim to minimize the least squares energy functional  
\begin{equation}
J_{\g d}(\g p):=\frac 12\Vert F(\g p)(\g t) - \g d\Vert_{\ell^2}^2.
\end{equation}
In this expression, the norm $\ell^2$ is the normalized euclidean norm defined by 
\begin{equation}
\Vert \g d \Vert_{\ell^2}^2=\frac{1}{n}\sum_{i=1}^n d_i^2.
\end{equation}
Denoting by $\tau$ the step of the subdivision $\g t$, we know using quadrature results (see \cite{davis1} for more details) that for all functions $f\in \text{H}^{2}(\Gamma_R)$,
\begin{equation}\label{control_l2}
\left| \Vert f\Vert_{\text{L}^2(\Gamma_R)}-\Vert f(\g t)\Vert_{\ell^2}\right|\leq 2\tau \Vert f\Vert_{\text{H}^1(\Gamma_R)}. 
\end{equation}
It shows that the $\ell^2$ norm is a good approximation of the $\text{L}^2(\Gamma_R)$ norm if the step of discretization $\tau$ is small enough. 

We assume the knowledge of an open set $U\subset \C^*\times(0,+\infty)\times\R$ of the form 
\begin{equation}
U:=B_\C(0,z_{\text{max}})\times (\alpha_{\min},\alpha_{\max})\times (\beta_{\min},\beta_{\max}), \quad z_{\max},\alpha_{\min},\alpha_{\max} \in \R_+^\star,\quad \beta_{\min},\beta_{\max} \in \R,
\end{equation}
containing the solution $\g p$.  The following proposition shows that the least squares problem is locally well-posed if the sampling size $n$ is large enough, and it quantifies the error on the recovered parameters. 

\begin{prop}\label{leastsq}
There exists $n_0\in \N^*$ such that if $n\geq n_0$, then for every $\g p_0\in U$ and $\g d_0:=F(p_0)(\g x)$, there exist $\eps>0$ and $U'\subset U$ such that for every $\g d\in \R^{n}$ satisfying $\Vert \g d-\g d_0\Vert_{\ell^2}< \eps$ then the functional $J_{\g d}$ is strictly convex on $U'$ and admits a unique minimizer denoted $\g p_{\text{LS}}=(z_{\text{LS}},\alpha_{\text{LS}},\beta_{\text{LS}})$. Moreover, there exists a constant $C_5>0$ depending only on $U$, $\g p_0$ and $n$ such that 
\begin{equation}
\Vert \g p_{\text{LS}}-\g p_0\Vert_{2}\leq C_5 \Vert \g d-\g d_0\Vert_{\ell^2}. 
\end{equation}
Finally, if we denote 
\begin{equation}
\Lambda : \left\{\begin{array}{rcl} B_2(\g d_0,\eps) & \rightarrow & \R \\ \g d &\mapsto& \frac{\beta_{\text{LS}}}{\alpha_{\text{LS}}}\end{array} \right. ,
\end{equation}
the operator which approach the value of $x^\star_k$, there exists a constant $C_6>0$ depending only on $U$, $\g p_0$ and $n$ such that 
\begin{equation}
|\Lambda(\g d)-\Lambda(\g d_0)|\leq C_6\Vert \g d-\g d_0\Vert_{\ell^2}.
\end{equation}
\end{prop}

The proof of this result is given in Appendix B. In this proof we can see that the choice of the data discretization points $\g t$ is important in order to improve the accuracy of the reconstruction. We illustrate that by choosing $\g p_0=(2+i,1.4,-2.8)$ and comparing the direct reconstruction of the parameters \eqref{meth11} with the least squares method for three different sets $\g t$:
\begin{itemize}
\item $\g t_1$ is a discretization of $[-6,-1]$ with 100 points, where the tail of the Airy function is located.
\item $\g t_2$ is a discretization of $[-2,6]$ with 100 points, where the Airy function varies.
\item $\g t_3$ is a discretization of $[-6,6]$ with 200 points. 
\end{itemize}
In Figure \ref{airycomp}, we represent the four reconstructions where $\g d = \g d_0+\cN$ and $\cN$ is a normal random additive noise of amplitude $0.3$. We see that reconstructions with $\g t_2$ and $\g t_3$ seems more accurate than the one with $\g t_1$ and the direct reconstruction. We detail this in Figure \ref{ercomp}, where we compare the reconstruction errors with different noise amplitudes. We clearly see that using a least squares method improves the reconstruction if $\g t$ is well-chosen.

\begin{figure}[h]
\begin{center}
\input{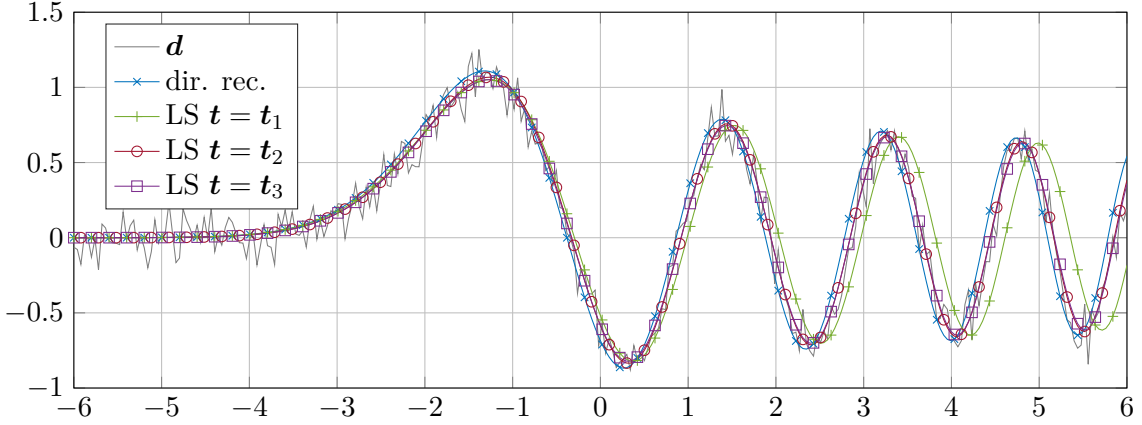}
\caption{\label{airycomp} Comparison of four different reconstructions of $F(\g p_0)$ where $\g d-\g d_0$ is a normal random additive noise of amplitude 0.3. Blue: direct reconstruction. Green: least square reconstruction with $\g t=\g t_1$. Red: least square reconstruction with $\g t=\g t_2$. Purple: least square reconstruction with $\g t=\g t_3$.}
\end{center}
\end{figure}

\begin{figure}[h]
\begin{center}
%
%
\definecolor{mycolor1}{rgb}{0.00000,0.44700,0.74100}%
\definecolor{mycolor2}{rgb}{0.4660 0.6740 0.1880}%
\definecolor{mycolor3}{rgb}{0.6350 0.0780 0.1840}%
\definecolor{mycolor4}{rgb}{0.49400,0.18400,0.55600}
\begin{tikzpicture}

\begin{loglogaxis}[%
width=10cm,
height=5cm,
at={(0in,0in)},
scale only axis,
xmin=0,
grid=minor,
xmax=1,
ymin=0,
ymax=0.34,
axis background/.style={fill=white},
ylabel={error on $\g p_{LS}$},
xlabel={$\Vert \g d-\g d_0\Vert_{\ell^2}/\Vert \g d_0\Vert_{\ell^2}$},
legend style={legend cell align=left, align=left, draw=white!15!black},
legend pos=north west,
]
\addplot [color=mycolor1,color=mycolor1,mark=x, mark size={2}]
  table[row sep=crcr]{%
0.05	0.0229091419009965\\
0.1	0.0296839032476869\\
0.15	0.0347675394993671\\
0.2	0.0395283993122679\\
0.25	0.0628113101305817\\
0.3	0.0819517079689435\\
0.35	0.0777339667472371\\
0.4	0.0843712736534832\\
0.45	0.0957897980154847\\
0.5	0.0931753649821667\\
0.55	0.131933905199308\\
0.6	0.122406842425223\\
0.65	0.172179152195183\\
0.7	0.165111352047512\\
0.75	0.172954573790091\\
0.8	0.188576692071903\\
0.85	0.20265548895462\\
0.9	0.160928507841859\\
0.95	0.191466725639646\\
1	0.190134483684818\\
};
\addlegendentry{dir. rec.}

\addplot [color=mycolor2, mark=+, mark size={2}]
  table[row sep=crcr]{%
0.05	0.014992189339849\\
0.1	0.0342105662184248\\
0.15	0.0433283728385579\\
0.2	0.0657137064865755\\
0.25	0.0799428237130169\\
0.3	0.0926905046838449\\
0.35	0.0913708230185353\\
0.4	0.101400163994652\\
0.45	0.129283815345702\\
0.5	0.147878645195776\\
0.55	0.182070659264914\\
0.6	0.180165126088297\\
0.65	0.210074112316479\\
0.7	0.206489977487207\\
0.75	0.259709749726001\\
0.8	0.260501086067901\\
0.85	0.266883825282437\\
0.9	0.276103623797072\\
0.95	0.231433585669496\\
1	0.311586100498883\\
};
\addlegendentry{LS $\g t=\g  t_1$}

\addplot [color=mycolor3, mark=o, mark size={2}]
  table[row sep=crcr]{%
0.05	0.00484929195048885\\
0.1	0.00906851707845026\\
0.15	0.012254152271014\\
0.2	0.0198246650523362\\
0.25	0.0268277025268608\\
0.3	0.0279021597240093\\
0.35	0.0282236515506412\\
0.4	0.0394460612108527\\
0.45	0.0435476548727766\\
0.5	0.033815092650431\\
0.55	0.0556446916725012\\
0.6	0.0483566485883064\\
0.65	0.0583038178349852\\
0.7	0.0622343193334649\\
0.75	0.0720811034918239\\
0.8	0.081368556405391\\
0.85	0.0690268998358236\\
0.9	0.08990620021772\\
0.95	0.0819579690696529\\
1	0.0919920663508916\\
};
\addlegendentry{LS $\g t=\g  t_2$}

\addplot [color=mycolor4, mark=square, mark size={2}]
  table[row sep=crcr]{%
0.05	0.00346665111091236\\
0.1	0.00751465065240954\\
0.15	0.0101095764445857\\
0.2	0.0138325245532753\\
0.25	0.0177118684533145\\
0.3	0.0209961154523075\\
0.35	0.0308169268930622\\
0.4	0.0336087464766161\\
0.45	0.0362292738581135\\
0.5	0.0371641494779487\\
0.55	0.0461517683496637\\
0.6	0.0440334149834422\\
0.65	0.0414759237500773\\
0.7	0.0474404901698698\\
0.75	0.0546005441414062\\
0.8	0.0566473001357812\\
0.85	0.0599101412815351\\
0.9	0.056259954122419\\
0.95	0.0765400102542804\\
1	0.0870218020203807\\
};
\addlegendentry{LS $\g t=\g  t_3$}

\end{loglogaxis}

\end{tikzpicture}%
 \hspace{3mm} 
%
%
\definecolor{mycolor1}{rgb}{0.00000,0.44700,0.74100}%
\definecolor{mycolor2}{rgb}{0.4660 0.6740 0.1880}%
\definecolor{mycolor3}{rgb}{0.6350 0.0780 0.1840}%
\definecolor{mycolor4}{rgb}{0.49400,0.18400,0.55600}
\begin{tikzpicture}

\begin{loglogaxis}[%
width=10cm,
height=5cm,
at={(0in,0in)},
scale only axis,
xmin=0,
xmax=1,
ymin=0,
ymax=0.08,
grid=minor,
axis background/.style={fill=white},
ylabel={error on $\Lambda(\g d)$},
xlabel={$\Vert \g d-\g d_0\Vert_{\ell^2}/\Vert \g d_0\Vert_{\ell^2}$},
legend style={legend cell align=left, align=left, draw=white!15!black},
legend pos=north west,
]

\addplot [color=mycolor1,color=mycolor1,mark=x, mark size={2}]
  table[row sep=crcr]{%
0.05	0.00354414789692381\\
0.1	0.00354414789692381\\
0.15	0.00666947497363536\\
0.2	0.00822073288174485\\
0.25	0.0207220411885911\\
0.3	0.0238473682653026\\
0.35	0.0209355064545601\\
0.4	0.0265343862847349\\
0.45	0.0320875893841144\\
0.5	0.0543461869320482\\
0.55	0.0459153118077599\\
0.6	0.049276915410933\\
0.65	0.0610465828779018\\
0.7	0.0619118224671191\\
0.75	0.073547891184748\\
0.8	0.0722329098031001\\
0.85	0.0603833758195022\\
0.9	0.0666226514475841\\
0.95	0.0655211353319053\\
1	0.0750192825762284\\
};
\addlegendentry{dir. rec.}

\addplot [color=mycolor2, mark=+, mark size={2}]
  table[row sep=crcr]{%
0.05	0.00238518344199795\\
0.1	0.00485142312434327\\
0.15	0.00927175437146035\\
0.2	0.0116174912355427\\
0.25	0.0137511010196861\\
0.3	0.0133877721564196\\
0.35	0.0182821589301048\\
0.4	0.0190527848673077\\
0.45	0.0228393991737257\\
0.5	0.027113741289705\\
0.55	0.0317961000519244\\
0.6	0.033092686831826\\
0.65	0.0316426748433885\\
0.7	0.0347121983940702\\
0.75	0.0452207255787851\\
0.8	0.0517463595170373\\
0.85	0.0472799728041939\\
0.9	0.0377977457324468\\
0.95	0.0434884907802953\\
1	0.0595625361408784\\
};
\addlegendentry{LS $\g t=\g  t_1$}

\addplot [color=mycolor3, mark=o, mark size={2}]
  table[row sep=crcr]{%
0.05	0.00128448075672102\\
0.1	0.00211688563726466\\
0.15	0.00320897978578655\\
0.2	0.00346445749810406\\
0.25	0.0063019740023234\\
0.3	0.00682193831635691\\
0.35	0.00843660710771867\\
0.4	0.0116505231959879\\
0.45	0.0122485312695519\\
0.5	0.0106418599404478\\
0.55	0.0121034220512995\\
0.6	0.0117036995281449\\
0.65	0.013847745396327\\
0.7	0.017306907721515\\
0.75	0.0152614836674537\\
0.8	0.0188059270322733\\
0.85	0.0150015606018552\\
0.9	0.0222407737705644\\
0.95	0.0207988997297874\\
1	0.0214219389521763\\
};
\addlegendentry{LS $\g t=\g  t_2$}

\addplot [color=mycolor4, mark=square, mark size={2}]
  table[row sep=crcr]{%
0.05	0.000784079447607485\\
0.1	0.0019154184728061\\
0.15	0.00242772902299415\\
0.2	0.00365513846516604\\
0.25	0.00476374840547367\\
0.3	0.00571128213029019\\
0.35	0.00670591819166667\\
0.4	0.00793636527009343\\
0.45	0.00853924884297249\\
0.5	0.00922043793957073\\
0.55	0.0115854192524997\\
0.6	0.0125349719111053\\
0.65	0.0128007196403181\\
0.7	0.0136579756279535\\
0.75	0.0139401797330852\\
0.8	0.0197095520908092\\
0.85	0.0148773165143485\\
0.9	0.0196882307538461\\
0.95	0.0183585399323251\\
1	0.0200480945662725\\
};
\addlegendentry{LS $\g t=\g  t_3$}

\end{loglogaxis}

\end{tikzpicture}%
\end{center}
\caption{\label{ercomp} Error of reconstruction of $\g p_0$ and $\Lambda(\g d_0)$ with respect to the noise on the data for different method of reconstruction. On the left, $\Vert p_{\text{LS}}-p_0\Vert_{2}$ with respect to $\Vert \g d-\g d_0\Vert_{\ell^2}/\Vert \g d\Vert_{\ell^2}$. On the right, $|\Lambda(\g d_0)-\Lambda(\g d)|$ with respect to $\Vert \g d-\g d_0\Vert_{\ell^2}/\Vert \g d\Vert_{\ell^2}$. }
\end{figure}
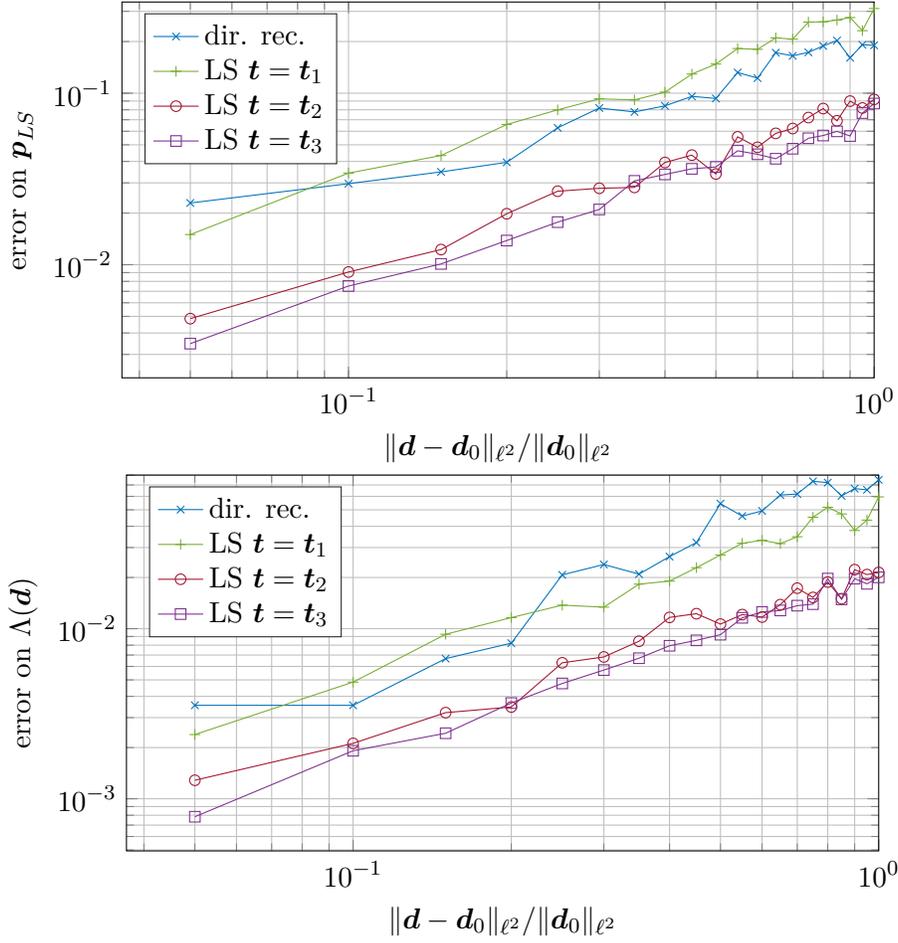

We now have a stable method to reconstruct an approximation $x^{\star,\app}_k$ of $x^\star_k$ from given measurements of $d(x):=u(x,0)$. Since $h(x^\star_k)=N\pi/k$, we can approximate the width at positions $x^{\star,\app}_k$ with the formula
\begin{equation}
h^\text{app}(x^{\star,\app}_k)=\frac{N\pi}{k}.
\end{equation}
This provides an approximation of $h$ in one point. 
Then, we change the frequency $k>0$ to get approximations of $h$ all along the support of $h'$. From now on, we denote $u_k$ the wavefield propagating at frequency $k$. As mentioned before, we assume that we already have an approximation of $\supp(h')$, $h_{\min}$ and $h_{\max}$ (see Appendix A). We denote 
\begin{equation}
k_{\max}:=\frac{N\pi}{h_{\min}}, \quad k_{\min}:=\frac{N\pi}{h_{\max}}.
\end{equation}
and we take a finite set of frequencies $K\subset (k_{\min},k_{\max})$.  For every frequency $k\in K$, we introduce $\g t_k$ a discretization of $\Gamma_R$ and we set
\begin{equation}
\forall k\in K\qquad \g d_k:=u_k(\g t_k,0)+\zeta_k,
\end{equation}
the measurement of $u_k$ corrupted by an error term $\zeta_k$. For every $k\in K$, using Proposition \ref{leastsq}, we have an $\eps_k$ and a constant $C_6^k$ associated with $\g d_{k,0}=d^\app_{z,\alpha,x^\star_k}$. We define
\begin{equation} \eps:=\min_{k\in K}(\eps_k), \qquad C_6:=\max_{k\in K}(C_6^k),
\end{equation}
and 
\begin{equation}
X^{\star,\app}:=\{\Lambda(\g d_k) ,\, k\in K\}.
\end{equation} 
Using the known approximation of $\supp(h')$ provided by Appendix A, we set the coordinates $x_0$ and $x_{n+1}$ such that $\text{supp} (h')\subset (x_0,x_{n+1})$ and $X^{\star,\app}\subset (x_0,x_{n+1})$. We then define the function $h^{\app}$ as the piecewise linear function such that 
\begin{equation}\label{happ}
 h^{\app}(x^{\star,\app}_{k})=\frac{N\pi}{k} \quad \forall\,  k\in K, \qquad  h^{\app}(x_0)=\frac{N\pi}{k_{\max}}, \qquad h^{\app}(x_{n+1})=\frac{N\pi}{k_{\min}}. 
\end{equation}

Using all the previous results, we are able to quantify the error of reconstruction between $h^{\app}$ and $h$:
\begin{theorem}\label{th3}
Let $K$ be a finite subset of $(k_{\min},k_{\max})$. Assume the same hypotheses than in Theorem \ref{th1}, and fix $N\in \N^*$ and $\theta>0$ such that for every $k\in K$,  $h'(x^\star_k)\geq \theta\eta$. There exist $\eta_1>0$ and $\zeta_0>0$ such that if $\eta <\eta_1$ and $\max_{k\in K}\Vert \zeta_k\Vert_{\ell^2}\leq \zeta_0$, then there exists a constant $C_7>0$ depending only on $h_{\min}$, $h_{\max}$, $N$, $r$, $\mu$ and $\theta$ such that 
\begin{equation}
 \Vert h^{\app}(X^{\star,\app})-h(X^{\star,\app})\Vert_\infty\leq \eta C_7\left(\eta^{1/3}+\max_{k\in K}\Vert \zeta_k\Vert_{\ell^2}\right). 
\end{equation}
\end{theorem}

\begin{proof}

For every $k\in K$, we notice that 
\begin{equation}\nonumber
|h^{\app}(x^{\star,\app}_k)-h(x^{\star,\app}_k)|=|h(x^\star_{k})-h(x^{\star,\app}_k)|\leq \eta |x^\star_{k}-x^{\star,\app}_k|=\eta|\Lambda(\g d_0)-\Lambda(\g d_{k,0})|.
\end{equation}
Using Theorem \ref{th2} combined with the control \eqref{control_l2}, we know that
\begin{equation}\nonumber
\Vert \g d_k-\g d_{0,k}\Vert_{\ell^2}\leq 2\eta^{1/3} \tau C_4  \left(\Vert d^\ex\Vert_{\text{H}^1(\Gamma_R)}+\Vert d^\app_{z,\alpha,x^\star_k}\Vert_{\text{H}^1(\Gamma_R)}\right).
\end{equation}
Using Theorem \ref{th1}, we can control both these norm by a constant $c_2>0$ which does no depend on $k$. Then, if $\eta$ and $\max_{k\in K}\Vert \zeta_k\Vert_{\ell^2}$ are small enough then,
\begin{equation}\nonumber
2\tau C_4c_2\eta^{1/3}+\max_{k\in K}\Vert \zeta_k\Vert_{\ell^2}\leq \eps,
\end{equation}
and using Proposition \ref{leastsq},
\begin{equation}\nonumber
|h^{\app}(x^{\star,\app}_k)-h(x^{\star,\app}_k)|\leq \eta C_6\left(2\tau C_4c_2\eta^{1/3}+\max_{k\in K}\Vert \zeta_k\Vert_{\ell^2}\right).
\end{equation}
\end{proof}

The first error term in this theorem is a consequence of the approximation of the data $d^\ex$ by the Airy function $d^\app$, while the second error term is caused by the eventual presence of measurements noises. We illustrate this reconstruction in Figure \ref{pointsxp} where we choose a subset $K$ with $20$ points, and we compare $h^{\app}(x^{\star,\app}_k)$ and the exact values $h(x^\star_k)$.

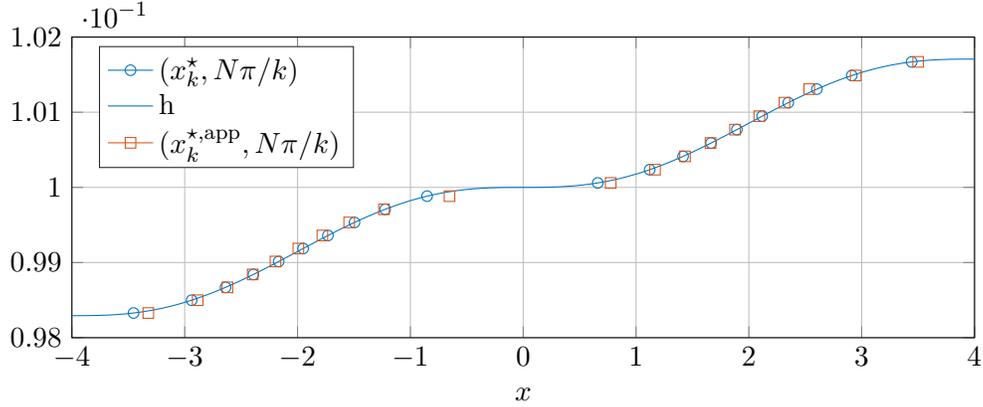
\begin{figure}[h]
\begin{center}
%
%
\definecolor{mycolor1}{rgb}{0.00000,0.44700,0.74100}%
\definecolor{mycolor2}{rgb}{0.85000,0.32500,0.09800}%
\begin{tikzpicture}

\begin{axis}[%
width=12cm,
height=4cm,
at={(0in,0in)},
scale only axis,
xmin=-4,
xmax=4,
grid=major, 
ymin=0.098,
legend pos=north west,
scaled y ticks=base 10:1,
xlabel={$x$},
ymax=0.102,
axis background/.style={fill=white},
legend style={legend cell align=left, align=left, draw=white!15!black}
]
\addplot [color=mycolor1, draw=none, mark=o, mark options={solid, mycolor1}]
  table[row sep=crcr]{%
3.44295382970716	0.101669665164718\\
2.91211164964376	0.101488158493932\\
2.60300353301074	0.101307298741015\\
2.34921961801124	0.101127082453547\\
2.11798956807277	0.100947506203629\\
1.89384234676567	0.10076856658767\\
1.66494226469951	0.100590260226165\\
1.41685552668203	0.100412583763489\\
1.12118126669493	0.100235533867684\\
0.660674568637204	0.100059107230251\\
-0.851939407222899	0.0998833005659405\\
-1.2237657830579	0.099708110612555\\
-1.49541249826153	0.0995335341307422\\
-1.73063372532789	0.099359567903797\\
-1.95055226459867	0.0991862087374644\\
-2.16729523867338	0.099013453459743\\
-2.39160351386512	0.0988412989206923\\
-2.63814483531309	0.0986697419922408\\
-2.93854449859568	0.0984987795679968\\
-3.45359757015163	0.0983284085630608\\
};
\addlegendentry{$(x^\star_k,N\pi/k)$}

\addplot [color=mycolor1]
  table[row sep=crcr]{%
-4	0.0982933333333333\\
-3.91919191919192	0.0982934698164101\\
-3.83838383838384	0.0982943919127539\\
-3.75757575757576	0.0982967955620121\\
-3.67676767676768	0.0983012826369804\\
-3.5959595959596	0.0983083650784094\\
-3.51515151515152	0.0983184690298115\\
-3.43434343434343	0.0983319389722671\\
-3.35353535353535	0.0983490418592318\\
-3.27272727272727	0.0983699712513428\\
-3.19191919191919	0.0983948514512255\\
-3.11111111111111	0.0984237416383004\\
-3.03030303030303	0.0984566400035896\\
-2.94949494949495	0.0984934878845235\\
-2.86868686868687	0.0985341738997473\\
-2.78787878787879	0.0985785380839283\\
-2.70707070707071	0.0986263760225616\\
-2.62626262626263	0.0986774429867776\\
-2.54545454545455	0.0987314580681482\\
-2.46464646464646	0.0987881083134937\\
-2.38383838383838	0.0988470528596894\\
-2.3030303030303	0.0989079270684723\\
-2.22222222222222	0.0989703466612474\\
-2.14141414141414	0.0990339118538952\\
-2.06060606060606	0.0990982114915775\\
-1.97979797979798	0.0991628271835444\\
-1.8989898989899	0.0992273374379414\\
-1.81818181818182	0.0992913217966152\\
-1.73737373737374	0.099354364969921\\
-1.65656565656566	0.0994160609715292\\
-1.57575757575758	0.0994760172532316\\
-1.49494949494949	0.0995338588397485\\
-1.41414141414141	0.0995892324635352\\
-1.33333333333333	0.0996418106995885\\
-1.25252525252525	0.0996912961002538\\
-1.17171717171717	0.0997374253300313\\
-1.09090909090909	0.0997799733003831\\
-1.01010101010101	0.0998187573045395\\
-0.92929292929293	0.0998536411523058\\
-0.848484848484849	0.0998845393048691\\
-0.767676767676768	0.0999114210096048\\
-0.686868686868687	0.0999343144348833\\
-0.606060606060606	0.0999533108048768\\
-0.525252525252525	0.0999685685343657\\
-0.444444444444445	0.0999803173635455\\
-0.363636363636364	0.0999888624928335\\
-0.282828282828283	0.0999945887176753\\
-0.202020202020202	0.0999979645633514\\
-0.121212121212121	0.0999995464197843\\
-0.0404040404040402	0.0999999826763445\\
0.0404040404040407	0.100000017323656\\
0.121212121212121	0.100000453580216\\
0.202020202020202	0.100002035436649\\
0.282828282828283	0.100005411282325\\
0.363636363636363	0.100011137507166\\
0.444444444444445	0.100019682636454\\
0.525252525252525	0.100031431465634\\
0.606060606060606	0.100046689195123\\
0.686868686868687	0.100065685565117\\
0.767676767676767	0.100088578990395\\
0.848484848484849	0.100115460695131\\
0.929292929292929	0.100146358847694\\
1.01010101010101	0.100181242695461\\
1.09090909090909	0.100220026699617\\
1.17171717171717	0.100262574669969\\
1.25252525252525	0.100308703899746\\
1.33333333333333	0.100358189300412\\
1.41414141414141	0.100410767536465\\
1.49494949494949	0.100466141160251\\
1.57575757575758	0.100523982746768\\
1.65656565656566	0.100583939028471\\
1.73737373737374	0.100645635030079\\
1.81818181818182	0.100708678203385\\
1.8989898989899	0.100772662562059\\
1.97979797979798	0.100837172816456\\
2.06060606060606	0.100901788508423\\
2.14141414141414	0.100966088146105\\
2.22222222222222	0.101029653338753\\
2.3030303030303	0.101092072931528\\
2.38383838383838	0.101152947140311\\
2.46464646464646	0.101211891686506\\
2.54545454545455	0.101268541931852\\
2.62626262626263	0.101322557013222\\
2.70707070707071	0.101373623977438\\
2.78787878787879	0.101421461916072\\
2.86868686868687	0.101465826100253\\
2.94949494949495	0.101506512115477\\
3.03030303030303	0.10154335999641\\
3.11111111111111	0.1015762583617\\
3.19191919191919	0.101605148548775\\
3.27272727272727	0.101630028748657\\
3.35353535353535	0.101650958140768\\
3.43434343434343	0.101668061027733\\
3.51515151515152	0.101681530970189\\
3.5959595959596	0.101691634921591\\
3.67676767676768	0.10169871736302\\
3.75757575757576	0.101703204437988\\
3.83838383838384	0.101705608087246\\
3.91919191919192	0.10170653018359\\
4	0.101706666666667\\
};
\addlegendentry{h}

\addplot [color=mycolor2, draw=none, mark=square, mark options={solid, mycolor2}]
  table[row sep=crcr]{%
3.5	0.101669665164718\\
2.94680614901442	0.101488158493932\\
2.53502632365852	0.101307298741015\\
2.31540963554168	0.101127082453547\\
2.09267349197802	0.100947506203629\\
1.87692049807251	0.10076856658767\\
1.66021515752679	0.100590260226165\\
1.43418851056657	0.100412583763489\\
1.16820757566861	0.100235533867684\\
0.775311044853517	0.100059107230251\\
-0.653049538280735	0.0998833005659405\\
-1.23377527151105	0.099708110612555\\
-1.54221274636732	0.0995335341307422\\
-1.77985748877314	0.099359567903797\\
-1.99284190069748	0.0991862087374644\\
-2.19562111753366	0.099013453459743\\
-2.39712586804073	0.0988412989206923\\
-2.62136478736978	0.0986697419922408\\
-2.88491718671545	0.0984987795679968\\
-3.3226927471198	0.0983284085630608\\
};
\addlegendentry{$(x^{\star,\text{app}}_k,N\pi/k)$}

\end{axis}
\end{tikzpicture}%
\end{center}
\caption{\label{pointsxp} Representation of $x^\star_k$ and $x^{\star,\app}_k$ when $K=\{30.9:31.95:20\}$. Here, $h$ is defined in \eqref{h2}, $N=1$, and data are generated as explained in section \ref{num} with a source $f(x,y)=\delta_{x=6}\varphi_1(y)$.}
\end{figure}

To conclude, we have proved in this section that we are able to reconstruct in a stable way a set of resonant positions $x^\star_k$, which leads to a stable reconstruction of the function $h$. We present in the next section the general idea needed to generalize this reconstruction to more realistic non-monochromatic sources terms. 

\section{Inversion using a general source term \label{section_general_source}}

We now consider the general case without any \text{a priori} assumptions on the source terms $f$ and $b$. We use the same arguments as before to reconstruct the width using the locally resonant mode $N$. However, other modes may also be present in the wavefield, and the previous approximation of $d^\ex$ by $u_N^\app$ (see \eqref{dex}) may not be valid. We provide here two methods to exploit the framework developed in the previous section. 

The first idea is to treat all resonant and evanescent modes as an added noise to the locally resonant mode. Indeed, we know using Theorem \ref{th1} that 
\begin{equation}
d^\ex:=u(x,0)\approx \sum_{n\in \N} \int_\R G_{n}^{\text{app}}(x,s)g_n(s)\dd s\,\ph_n\left(y\right).
\end{equation}
Defining the noise
\begin{equation}
\zeta_k(x):=\sum_{n\neq N}\int_\R G_{n}^{\text{app}}(x,s)g_n(s)\,\dd s\ph_n\left(y\right), 
\end{equation}
we see that 
\begin{equation}
d^\ex\approx u_{N}^{\app}(x)\varphi_N(0)+\zeta_k(x).
\end{equation}
To make a successful use of Theorem \ref{th3}, the noise $\zeta_k$ needs to be smaller than the parameter $\zeta_0$. Using the expression of $G_{n}^{\app}$ given in \eqref{greenfunction}, we notice that 
\begin{equation}
\forall n\in \N, \qquad |G_{n}^{\app}(x)|=\cO\left(\frac{1}{\min(|k_n|)}\right).
\end{equation}
Since
\begin{equation}
\frac{|k_N|}{|k_n|}\leq\frac{|k_N|^2h_{\max}^2}{\sqrt{2N-1}\pi},
\end{equation}
the amplitude of $\zeta_k$ seems to be neglectable compared to the amplitude of $u_{N}^{\app}$ if the width $h_{\max}$ is small enough and if $\Vert f_n\Vert_{\text{L}^2(\R)}=\cO( \Vert f_N\Vert_{\text{L}^2(\R)})$ for every $n\neq N$. We illustrate this in Figure \ref{ampli} where we compare the amplitude of $u(x,0)$ with the amplitude of $u_{N}^{\app}(x)\varphi_N(0)$ for source terms $f(x,y)=y^2\delta_{x=6}$ and $b_{\top}(x)=\delta_{x=6}$. We see that it is possible to fit directly the Airy function on $u(x,0)$, since the noise is very small. 

\begin{figure}[h]
\begin{center}
\input{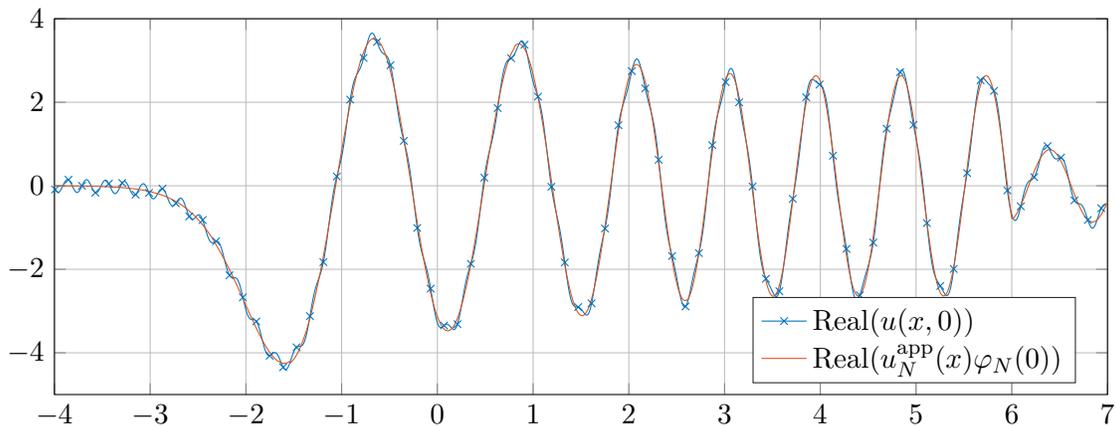}
\end{center}
\caption{\label{ampli} Representation of $u(x,0)$ and $u_{N}^{\app}(x)\varphi_N(0)$ generated by a source $f(x,y)=y^2\delta_{x=6}$ and $b_{\top}(x)=\delta_{x=6}$. Here $N=1$, $k=31.7$ and $h$ is defined as in \eqref{h2}. }
\end{figure}

However, in the general case, the noise of the other modes can not always be neglected. In this case, we can filter the measurements to keep only the locally resonant part. Evanescent modes vanish away from the source and so their contribution is negligible in $u(x,0)$. As for propagative modes, we notice that $k_n(x)$ is almost constant when $n<N$. Using \eqref{greenfunction}, it means that every propagative mode is oscillating with a frequency almost equal to $k_n$. We illustrate it in Figure \ref{fourier} with the Fourier transform of $u(x,0)$ and the contribution of each mode. Using a filter cutting all frequencies around $k_n(x)$ for $n<N$, we can clean the signal and get a good approximation of $u_{N}(x)\varphi_N(0)$ (see \cite{mallat1} for more details). We illustrate this in Figure \ref{passbas} where we represent $u(x,0)$, $u_{N}(x)\varphi_n(0)$ and the approximation obtained using this Fourier filtering.

\begin{figure}[h]
\begin{center}
\input{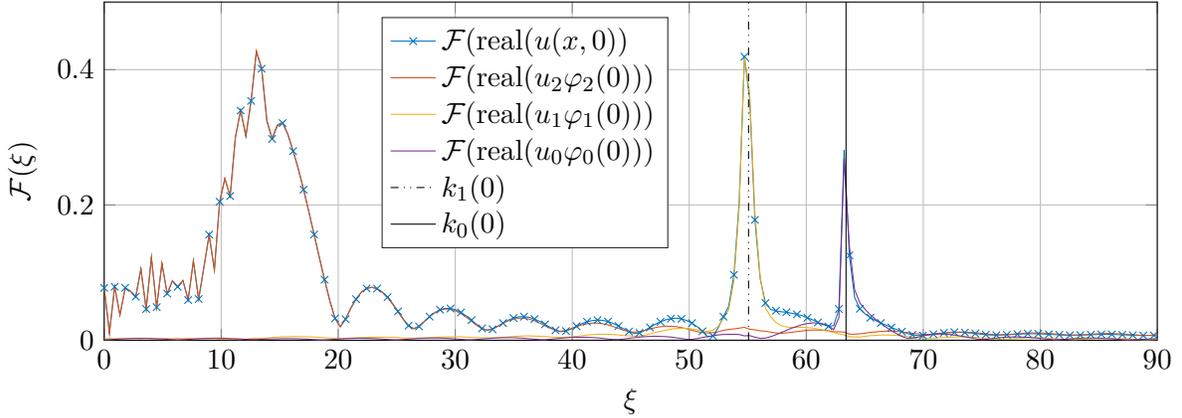}
\end{center}
\caption{\label{fourier} Influence of each mode on the Fourier transform of the measurements real part. Fourier transform of $\text{Real}(u_{n}\varphi_n(0))$ are plotted for every non evanescent mode ($n=0,1,2$). For comparison purpose, $k_n(0)$ is represented for every propagative modes ($n=1,2$). Here, $k=63.4$, $N=2$, $f(x,y)=\delta_{x=6}(3\varphi_0(y)+2\varphi_1(y)+\varphi_2(y))$ and $b_\top=\delta_{x=6}$.}
\end{figure}

\begin{figure}[h]
\begin{center}
\input{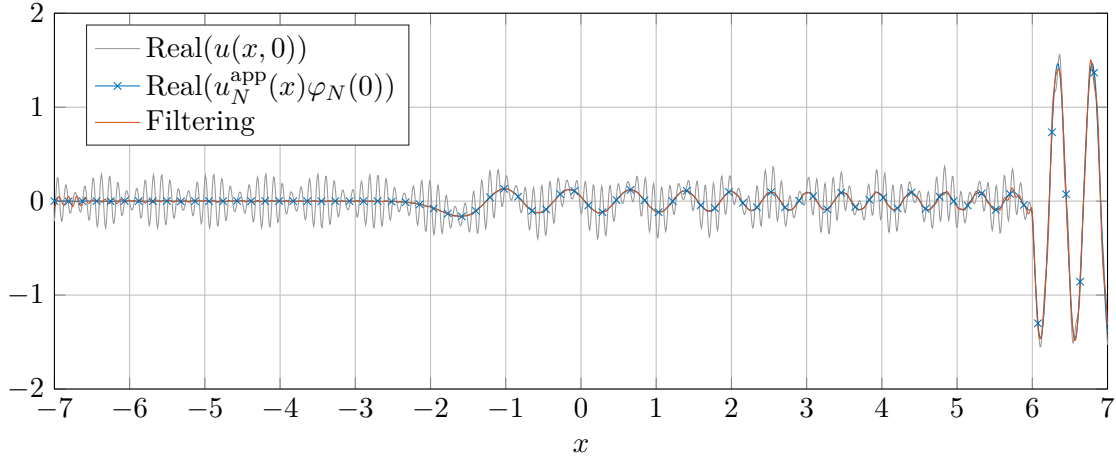}
\end{center}
\caption{\label{passbas} Fourier filtering of $\text{Real}(u_k(x,0))$ and comparison with $\text{Real}(u_{k,N}^{\app}(x)\varphi_N(0))$. Here, $k=63.4$, $N=2$, $f(x,y)=\delta_{x=6}(3\varphi_0(y)+2\varphi_1(y)+\varphi_2(y))$ and $b_\top=\delta_{x=6}$.}
\end{figure}

We can now fit the three parameters Airy function on the measurements, and find an approximation of $x^\star_k$ as before. It proves that our method can be used in the case of general sources, providing a filtering of propagative mode, and does not need any \emph{a priori} information on sources as long as their locally resonant part does not vanish. Using the extension of Theorem \ref{th1} to non monotone waveguides provided in Section 4 of \cite{bonnetier2}, all the results presented in this section remain true in any kind of slowly varying waveguide.

\section{Numerical computations}

In this section, we show some numerical applications of our reconstruction method on slowly varying waveguides. We simulate one side boundary measurements using numerically generated data, and we provide reconstructions of increasing an non monotone waveguides with different shape profiles. 

\subsection{Generation of data \label{num}}
In the following, numerical solutions of \eqref{eqmatlab} are generated using the software Matlab to solve numerically the equation in the waveguide $\Omega$. In every numerical simulation, we assume that $h'$ is supported between $x=-7$ and $x=7$. To generate the solution $u$ of \eqref{eqmatlab} on $\Omega_7$, we use a self-coded finite element method and a perfectly matched layer \cite{berenger1} on the left side of the waveguide between $x=-15$ and $x=-8$ and on the right side between $x=8$ and $x=15$. The coefficient of absorption for the perfectly matched layer is defined as $\alpha=-k((x-8)\textbf{1}_{x\geq 8}-(x+8)\textbf{1}_{x\leq -8})$ and $k^2$ is replaced in the Helmholtz equation with $k^2+i\alpha$. The structured mesh is built with a stepsize of $10^{-3}$. 

\subsection{Method of reconstruction}
Using all the previous results, we summarize here all the steps to reconstruct the approximation $h^{\app}$ of the width $h$. 
\begin{enumerate}
\item Find an approximation of $\supp(h')$, $k_{\min}$ and $k_{\max}$ using the method described at the beginning of Appendix A.
\item Choose a set of frequencies $K\subset (k_{\min},k_{\max})$ and sources $f$, $b_{\top}$, $b_{\bot}$. Then, for every frequency, measure the wavefield $d(x):=u(x,0)$ solution of \eqref{eqmatlab}. 
\item Filter the data by eliminating in the Fourier transform responses around $k_n$ for every a propagative mode $n$. 
\item Find an approximation of the coordinate $x_{\max}$ where $|u(x,0)|$ is maximal. Then, choose a length $R>0$, and discretize the interval $[x_{\max}-R,x_{\max}+R]$ into the set $\g t_k$. The available data $\g t_k$ are then the measurements of $u(\g t_k,0)$.
\item For every frequency, minimize the quantity $\Vert \g d_k-F(\g p)(\g t_k)\Vert_{\ell^2}$ using a gradient descent to find the approximation $x^{\star,\app}_k$ of $x^\star_k$. The direct reconstruction method can be used to initialize the gradient descent method. 
\item Compute $h^{\app}$ using expression \eqref{happ}.
\end{enumerate}

In step $3$, we should normally discretize the interval $\Gamma_R(x^\star_k)$. However, since we do not know yet the value of $x^\star_k$, we use the fact that $x^\star_k$ is located near $x_{\max}$, and that the interval $[x_{\max}-R,x_{\max}+R]$ can be included into a bigger interval $\Gamma_{R'}(x^\star_k)$ where $R'>R$.

\subsection{Numerical results}

We now apply this method to reconstruct different profiles of slowly varying waveguides. Firstly, we present in Figure \ref{increasing} the reconstruction $h^{\app}$ obtained for different increasing functions $h$. We choose four different waveguide profiles
\begin{equation}\label{h1}
\Omega_1: \quad h_1(x)=0.1+\gamma_1\left(\frac{x^5}{5}-32\frac{x^3}{3}+256x\right)\textbf{1}_{-4\leq x\leq 4}-\gamma_2\textbf{1}_{x<-4}+\gamma_2\textbf{1}_{x>4},
\end{equation}
\begin{equation}\label{h2}
\Omega_2:\quad h_2(x)=0.1+\gamma_3\left(\frac{x^5}{5}-2x^4+16\frac{x^3}{3}\right)\left(\textbf{1}_{0\leq x\leq 4}-\textbf{1}_{-4\leq x<0}\right)+ \gamma_4\left(\textbf{1}_{x>4}-\textbf{1}_{x<-4}\right),
\end{equation}
\begin{equation}\label{h3}
\Omega_3 :\quad h_3(x)=0.1+\gamma_5 x\textbf{1}_{-4\leq x\leq 4}+4\gamma_5\textbf{1}_{x>4}-4\gamma_4\textbf{1}_{x<-4},
\end{equation}
\begin{equation}\label{h4}
\Omega_4:\quad h_4(x)=0.1-4\gamma_5+4\gamma_5\frac{\sqrt{x+4}}{\sqrt{2}}\,\textbf{1}_{-4\leq x\leq 4}+8\gamma_5\textbf{1}_{x>4}.
\end{equation}
where $\gamma_1=3.10^{-6}$, $\gamma_2=8192/5.10^{-6}$, $\gamma_3=5.10^{-5}$, $\gamma_4=53/3.10^{-5}$, $\gamma_5=0.01/30$. All these profiles are represented in black in Figure \ref{increasing}. We impose that sources of \eqref{eqmatlab} need to be located at the right of the waveguide in order to generate significant locally resonant mode (see Figure \ref{direct_multi_freq2}), and we define
\begin{equation}\label{sourceincr}
f(x,y)=y\delta_{x=6}, \qquad b_{\top}(x)=\delta_{x=6}, \qquad b_{\bot}=0.
\end{equation}
We choose to work with the following sets of frequencies
\begin{equation}\label{Ki}
K_1=\{30.92:31.93:20\}, \quad K_2=\{30.9:31.95:20\}, \quad  K_3=K_4=\{31.01:31.83:20\}.
\end{equation}

The profile $h_1$ and the set $K_1$ satisfy all the assumptions of Theorem \ref{th2}, while the derivative of $h_2\in \mathcal{C}^2(\R)$ vanishes once. The last two profiles $h_3$ and $h_4$ are not in $\mathcal{C}^2(\R)$ and show corners where the derivative of the profile is not continuous. The function $h_3'$ is piecewise continuous and bounded, contrary to $h_4'$  which explodes at $x=-4$. 

We plot in red in Figure \ref{increasing} the reconstructions $h^{\app}$, slightly shifted, obtained using our method of reconstruction. We also compute the $L^\infty$ norm of $h-h^{\app}$ in each situation. We clearly see that the reconstruction deteriorates when $h'(x^\star_k)$ is too small or when the function $h$ is not sufficiently smooth. 

\begin{figure}[h]
\begin{center}
\scalebox{0.8}{\input{rec6}}\hspace{5mm} \scalebox{0.8}{
%
%
\begin{tikzpicture}

\begin{axis}[%
width=3.5in,
height=1in,
at={(0in,0in)},
scale only axis,
xmin=-5,
xmax=5,
ymin=0.092,
hide axis, 
ymax=0.102006666666667,
axis background/.style={fill=white},
legend style={legend cell align=left, align=left, draw=white!15!black}
]
\addplot [color=black]
  table[row sep=crcr]{%
-7	0.0982933333333333\\
-6.85858585858586	0.0982933333333333\\
-6.71717171717172	0.0982933333333333\\
-6.57575757575758	0.0982933333333333\\
-6.43434343434343	0.0982933333333333\\
-6.29292929292929	0.0982933333333333\\
-6.15151515151515	0.0982933333333333\\
-6.01010101010101	0.0982933333333333\\
-5.86868686868687	0.0982933333333333\\
-5.72727272727273	0.0982933333333333\\
-5.58585858585859	0.0982933333333333\\
-5.44444444444444	0.0982933333333333\\
-5.3030303030303	0.0982933333333333\\
-5.16161616161616	0.0982933333333333\\
-5.02020202020202	0.0982933333333333\\
-4.87878787878788	0.0982933333333333\\
-4.73737373737374	0.0982933333333333\\
-4.5959595959596	0.0982933333333333\\
-4.45454545454546	0.0982933333333333\\
-4.31313131313131	0.0982933333333333\\
-4.17171717171717	0.0982933333333333\\
-4.03030303030303	0.0982933333333333\\
-3.88888888888889	0.0982936840590019\\
-3.74747474747475	0.0982972311626275\\
-3.60606060606061	0.0983073224717585\\
-3.46464646464646	0.0983264747021903\\
-3.32323232323232	0.0983564338179318\\
-3.18181818181818	0.0983982428961426\\
-3.04040404040404	0.0984523099920701\\
-2.8989898989899	0.0985184760039864\\
-2.75757575757576	0.098596082538126\\
-2.61616161616162	0.0986840397736227\\
-2.47474747474747	0.0987808943274465\\
-2.33333333333333	0.0988848971193416\\
-2.19191919191919	0.0989940712367626\\
-2.05050505050505	0.0991062797998125\\
-1.90909090909091	0.0992192938261793\\
-1.76767676767677	0.0993308600960737\\
-1.62626262626263	0.0994387690171657\\
-1.48484848484848	0.0995409224895225\\
-1.34343434343434	0.0996354017705449\\
-1.2020202020202	0.099720535339905\\
-1.06060606060606	0.0997949667644835\\
-0.91919191919192	0.0998577225633063\\
-0.777777777777778	0.0999082800724822\\
-0.636363636363637	0.0999466353101399\\
-0.494949494949495	0.0999733708413651\\
-0.353535353535354	0.099989723643138\\
-0.212121212121212	0.09999765296927\\
-0.0707070707070709	0.0999999082153415\\
0.0707070707070709	0.100000091784659\\
0.212121212121212	0.10000234703073\\
0.353535353535354	0.100010276356862\\
0.494949494949495	0.100026629158635\\
0.636363636363637	0.10005336468986\\
0.777777777777778	0.100091719927518\\
0.91919191919192	0.100142277436694\\
1.06060606060606	0.100205033235517\\
1.2020202020202	0.100279464660095\\
1.34343434343434	0.100364598229455\\
1.48484848484848	0.100459077510478\\
1.62626262626263	0.100561230982834\\
1.76767676767677	0.100669139903926\\
1.90909090909091	0.100780706173821\\
2.05050505050505	0.100893720200188\\
2.19191919191919	0.101005928763237\\
2.33333333333333	0.101115102880658\\
2.47474747474747	0.101219105672553\\
2.61616161616162	0.101315960226377\\
2.75757575757576	0.101403917461874\\
2.8989898989899	0.101481523996014\\
3.04040404040404	0.10154769000793\\
3.18181818181818	0.101601757103857\\
3.32323232323232	0.101643566182068\\
3.46464646464647	0.10167352529781\\
3.60606060606061	0.101692677528242\\
3.74747474747475	0.101702768837372\\
3.88888888888889	0.101706315940998\\
4.03030303030303	0.101706666666667\\
4.17171717171717	0.101706666666667\\
4.31313131313131	0.101706666666667\\
4.45454545454546	0.101706666666667\\
4.5959595959596	0.101706666666667\\
4.73737373737374	0.101706666666667\\
4.87878787878788	0.101706666666667\\
5.02020202020202	0.101706666666667\\
5.16161616161616	0.101706666666667\\
5.3030303030303	0.101706666666667\\
5.44444444444444	0.101706666666667\\
5.58585858585859	0.101706666666667\\
5.72727272727273	0.101706666666667\\
5.86868686868687	0.101706666666667\\
6.01010101010101	0.101706666666667\\
6.15151515151515	0.101706666666667\\
6.29292929292929	0.101706666666667\\
6.43434343434343	0.101706666666667\\
6.57575757575758	0.101706666666667\\
6.71717171717172	0.101706666666667\\
6.85858585858586	0.101706666666667\\
7	0.101706666666667\\
};

\addplot [color=black]
  table[row sep=crcr]{%
-7	0.092\\
7	0.092\\
};

\addplot [color=red]
  table[row sep=crcr]{%
7	0.102006666666667\\
4	0.102006666666667\\
3.4	0.101969665164718\\
2.94626062729308	0.101788158493932\\
2.53473482654324	0.101607298741015\\
2.31529753323391	0.101427082453547\\
2.09189548788103	0.101247506203629\\
1.86329821637842	0.10106856658767\\
1.65360212915866	0.100890260226165\\
1.4288794947909	0.100712583763489\\
1.16766662294117	0.100535533867684\\
0.78362706150052	0.100359107230251\\
-0.453092004208764	0.100183300565941\\
-1.23413938374612	0.100008110612555\\
-1.54266569749724	0.0998335341307422\\
-1.77958329298222	0.099659567903797\\
-1.99289787057962	0.0994862087374644\\
-2.19518483535566	0.099313453459743\\
-2.40547492125573	0.0991412989206923\\
-2.62074337405034	0.0989697419922408\\
-2.88689939735507	0.0987987795679968\\
-3.31565557885294	0.0986284085630608\\
-4	0.0985933333333333\\
-7	0.0985933333333333\\
};

\end{axis}
\end{tikzpicture}%
}\\[4mm]
\scalebox{0.8}{\input{rec4}} \hspace{5mm} \scalebox{0.8}{
%
%
\begin{tikzpicture}

\begin{axis}[%
width=3.5in,
height=1in,
at={(0in,0in)},
scale only axis,
xmin=-5,
xmax=5,
ymin=0.092,
ymax=0.101609018174453,
axis background/.style={fill=white},
axis x line*=bottom,
axis y line*=left,
hide axis,
legend style={legend cell align=left, align=left, draw=white!15!black}
]

\addplot [color=black]
  table[row sep=crcr]{%
-5	0.0986666666666667\\
-4.8989898989899	0.0986666666666667\\
-4.7979797979798	0.0986666666666667\\
-4.6969696969697	0.0986666666666667\\
-4.5959595959596	0.0986666666666667\\
-4.49494949494949	0.0986666666666667\\
-4.39393939393939	0.0986666666666667\\
-4.29292929292929	0.0986666666666667\\
-4.19191919191919	0.0986666666666667\\
-4.09090909090909	0.0986666666666667\\
-3.98989898989899	0.0987614225406025\\
-3.88888888888889	0.098980936347194\\
-3.78787878787879	0.0991008926315193\\
-3.68686868686869	0.0991942450447354\\
-3.58585858585859	0.0992734002997043\\
-3.48484848484848	0.0993433589585634\\
-3.38383838383838	0.0994067337003642\\
-3.28282828282828	0.0994650938523382\\
-3.18181818181818	0.0995194695320891\\
-3.08080808080808	0.0995705800937857\\
-2.97979797979798	0.0996189514140923\\
-2.87878787878788	0.0996649817455035\\
-2.77777777777778	0.0997089812799608\\
-2.67676767676768	0.0997511972146642\\
-2.57575757575758	0.0997918303285369\\
-2.47474747474747	0.0998310463394739\\
-2.37373737373737	0.0998689839205145\\
-2.27272727272727	0.0999057605030461\\
-2.17171717171717	0.0999414765708512\\
-2.07070707070707	0.0999762188986377\\
-1.96969696969697	0.100010063035842\\
-1.86868686868687	0.100043075240153\\
-1.76767676767677	0.100075314002809\\
-1.66666666666667	0.100106831266313\\
-1.56565656565657	0.100137673407143\\
-1.46464646464646	0.100167882036619\\
-1.36363636363636	0.100197494659432\\
-1.26262626262626	0.100226545219596\\
-1.16161616161616	0.100255064556473\\
-1.06060606060606	0.100283080788355\\
-0.959595959595959	0.100310619637176\\
-0.858585858585859	0.100337704705043\\
-0.757575757575758	0.100364357711043\\
-0.656565656565657	0.100390598695086\\
-0.555555555555555	0.100416446194225\\
-0.454545454545454	0.100441917395864\\
-0.353535353535354	0.100467028271447\\
-0.252525252525253	0.100491793693594\\
-0.151515151515151	0.100516227539111\\
-0.0505050505050502	0.100540342779907\\
0.0505050505050502	0.100564151563519\\
0.151515151515151	0.100587665284635\\
0.252525252525253	0.100610894648846\\
0.353535353535354	0.100633849729613\\
0.454545454545454	0.100656540019319\\
0.555555555555555	0.100678974475142\\
0.656565656565657	0.100701161560369\\
0.757575757575758	0.100723109281708\\
0.858585858585859	0.100744825223036\\
0.959595959595959	0.100766316576024\\
1.06060606060606	0.100787590167955\\
1.16161616161616	0.100808652487061\\
1.26262626262626	0.100829509705637\\
1.36363636363636	0.10085016770117\\
1.46464646464646	0.100870632075681\\
1.56565656565657	0.100890908173469\\
1.66666666666667	0.100911001097414\\
1.76767676767677	0.100930915723968\\
1.86868686868687	0.100950656716983\\
1.96969696969697	0.100970228540461\\
2.07070707070707	0.100989635470342\\
2.17171717171717	0.101008881605417\\
2.27272727272727	0.101027970877428\\
2.37373737373737	0.101046907060453\\
2.47474747474747	0.101065693779616\\
2.57575757575758	0.101084334519189\\
2.67676767676768	0.101102832630137\\
2.77777777777778	0.101121191337153\\
2.87878787878788	0.101139413745217\\
2.97979797979798	0.10115750284573\\
3.08080808080808	0.101175461522247\\
3.18181818181818	0.101193292555839\\
3.28282828282828	0.101210998630128\\
3.38383838383838	0.101228582335994\\
3.48484848484848	0.101246046176006\\
3.58585858585859	0.101263392568577\\
3.68686868686869	0.101280623851876\\
3.78787878787879	0.101297742287503\\
3.88888888888889	0.101314750063958\\
3.98989898989899	0.101331649299906\\
4.09090909090909	0.101333333333333\\
4.19191919191919	0.101333333333333\\
4.29292929292929	0.101333333333333\\
4.39393939393939	0.101333333333333\\
4.49494949494949	0.101333333333333\\
4.5959595959596	0.101333333333333\\
4.6969696969697	0.101333333333333\\
4.7979797979798	0.101333333333333\\
4.8989898989899	0.101333333333333\\
5	0.101333333333333\\
};

\addplot [color=red]
  table[row sep=crcr]{%
4.15280508030634	0.101609018174453\\
3.2947429719728	0.101468218196651\\
2.49818988813794	0.101327809045252\\
1.76363989947096	0.101187789095252\\
1.06075765337634	0.101048156730646\\
0.418194869046171	0.100908910344361\\
-0.177702021882812	0.100770048338197\\
-0.732886730813372	0.100631569122764\\
-1.24117895070901	0.100493471117425\\
-1.70380783960414	0.100355752750232\\
-2.13500132448921	0.100218412457869\\
-2.50647201385863	0.100081448685589\\
-2.83919688290037	0.099944859887161\\
-3.1295055442451	0.099808644524808\\
-3.38117118301614	0.0996728010691496\\
-3.59450084312227	0.0995373279991456\\
-3.79696295565863	0.0994022238020389\\
-3.94918808128174	0.0992674869732994\\
-4.10606783275681	0.0991331160165677\\
-4.30760348431977	0.0989991094436002\\
};

\addplot [color=red]
  table[row sep=crcr]{%
5	0.101609018174453\\
4.15280508030634	0.101609018174453\\
};

\addplot[color=black] 
table[row sep=crcr]{
-7 0.092 \\
7 0.092 \\};

\addplot [color=red]
  table[row sep=crcr]{%
-5	0.0989991094436002\\
-4.30760348431977	0.0989991094436002\\
};

\end{axis}

\end{tikzpicture}%
}
\end{center}
\caption{\label{increasing} Reconstruction of four different increasing profiles. In black, the initial shape of $\Omega_5$ (not scaled), and in red the reconstruction, slightly shifted for comparison purposes. In each case, $K=K_i$ is defined in \eqref{Ki}, $h=h_i$ is defined in \eqref{h1}, \eqref{h2}, \eqref{h3}, \eqref{h4}, and the sources of \eqref{eqmatlab} are defined in \eqref{sourceincr}. We also compute the relative error $\Vert h-h^{\app}\Vert_\infty/h_{\max}$ between the exact reconstruction and the approximation. Top left: $i=1$, $\Vert h-h^{\app}\Vert_\infty/h_{\max}=0.49 \%$. Top right: $i=2$, $\Vert h-h^{\app}\Vert_\infty/h_{\max}=0.94 \%$. Bottom left: $i=3$, $\Vert h-h^{\app}\Vert_\infty/h_{\max}=0.40 \%$. Bottom right: $i=4$, $\Vert h-h^{\app}\Vert_\infty/h_{\max}=1.6 \%$.}
\end{figure}

Secondly, we present in Figure \ref{recgene} the reconstruction $h^{\app}$ obtained for different non monotoneous widths. We choose three different profiles defined by 
\begin{equation}\label{h5}
\Omega_5: \quad h_5(x)=0.1+\gamma_6\sin\left(\frac{\pi}{10}(x+5)\right)\textbf{1}_{-5\leq x\leq 5},
\end{equation}
\begin{equation}\label{h6}
\Omega_6: \quad h_6(x)=0.1-\gamma_7(x+5)\textbf{1}_{-5\leq x\leq 0} +\frac{\gamma_6}{4}(x-4)\textbf{1}_{0<x\leq 4},
\end{equation}
\begin{equation}\label{h7}
\Omega_7: \quad h_7(x)=0.1-\gamma_8\sqrt{3}+2\gamma_8\sin\left(\frac{4\pi\sqrt{x+5}}{3}\right)\textbf{1}_{-3.5\leq x\leq 4}+2\gamma_8\sin\left(\frac{4\pi\sqrt{1.5}}{3}\right)\textbf{1}_{x<-3.5},
\end{equation}
where $\gamma_6=25.10^{-4}$, $\gamma_7=5.10^{-4}$, $\gamma_8=4.10^{-4}$. All these profiles are represented in black in Figure \ref{recgene}. The profile $h_5$ represent a dilation of the waveguide, while $h_6$ represent a compression of the waveguide. The profile $h_7$ is the more general one with both compressions and dilations. Again, sources in \eqref{eqmatlab} are located in every large area of the waveguide to generate significant locally resonant modes and are defined by
\begin{equation}\label{f5}
f^5(x,y)=\delta_{x=0}y, \quad b_{\top}^5(x)=\delta_{x=0}(x), \quad b_{\bot}^5=0,
\end{equation}
\begin{equation}\label{f6}
f^6(x,y)=(\delta_{x=6}+\delta_{x=-6})y, \quad b_{\top}^6(x)=(\delta_{x=6}+\delta_{x=-6}), \quad b_{\bot}^6=0,
\end{equation}
\begin{equation}\label{f7}
f^7(x,y)=(\delta_{x=-1.5}+\delta_{x=6})y, \quad b_{\top}^7(x)=(\delta_{x=-1.5}+\delta_{x=6}), \quad b_{\bot}^7=0.
\end{equation}
We also define the frequency sets
\begin{equation}\label{Ki2}
K_5=\{30.65:31.4:20\}, \quad K_6=\{31.42:32.21:20\}, \quad K_7=\{30.97:31.43:20\}. 
\end{equation}
We plot in red in Figure \ref{recgene} the reconstructions $h^{\app}$, slightly shifted, obtained using our method of reconstruction. We also compute the $L^\infty$ norm of $h-h^{\app}$ in each situation.

\begin{figure}[h]
\begin{center}
\scalebox{0.8}{
%
%
\begin{tikzpicture}

\begin{axis}[%
width=3.5in,
height=1in,
at={(0in,0in)},
scale only axis,
xmin=-6,
xmax=6,
ymin=0.0929,
ymax=0.104,
hide axis,
axis background/.style={fill=white},
legend style={legend cell align=left, align=left, draw=white!15!black}
]

\addplot [color=black]
  table[row sep=crcr]{%
-6	0.1\\
-5.87878787878788	0.1\\
-5.75757575757576	0.1\\
-5.63636363636364	0.1\\
-5.51515151515152	0.1\\
-5.39393939393939	0.1\\
-5.27272727272727	0.1\\
-5.15151515151515	0.1\\
-5.03030303030303	0.1\\
-4.90909090909091	0.100071390126984\\
-4.78787878787879	0.100166476329958\\
-4.66666666666667	0.100261321158169\\
-4.54545454545454	0.100355787095683\\
-4.42424242424242	0.100449737175922\\
-4.3030303030303	0.100543035180251\\
-4.18181818181818	0.100635545835484\\
-4.06060606060606	0.100727135010013\\
-3.93939393939394	0.100817669908294\\
-3.81818181818182	0.100907019263377\\
-3.6969696969697	0.10099505352724\\
-3.57575757575758	0.101081645058616\\
-3.45454545454545	0.101166668308064\\
-3.33333333333333	0.10125\\
-3.21212121212121	0.101331519311437\\
-3.09090909090909	0.101411108047167\\
-2.96969696969697	0.101488650811132\\
-2.84848484848485	0.101564035173737\\
-2.72727272727273	0.101637151834863\\
-2.60606060606061	0.101707894782342\\
-2.48484848484848	0.101776161445663\\
-2.36363636363636	0.101841852844691\\
-2.24242424242424	0.101904873733177\\
-2.12121212121212	0.101965132736857\\
-2	0.102022542485937\\
-1.87878787878788	0.10207701974177\\
-1.75757575757576	0.102128485517542\\
-1.63636363636364	0.102176865192799\\
-1.51515151515152	0.102222088621637\\
-1.39393939393939	0.102264090234408\\
-1.27272727272727	0.102302809132787\\
-1.15151515151515	0.102338189178075\\
-1.03030303030303	0.102370179072587\\
-0.909090909090909	0.102398732434036\\
-0.787878787878788	0.102423807862777\\
-0.666666666666667	0.102445369001835\\
-0.545454545454546	0.102463384589619\\
-0.424242424242424	0.102477828505251\\
-0.303030303030303	0.102488679806433\\
-0.181818181818182	0.102495922759815\\
-0.0606060606060606	0.102499546863809\\
0.0606060606060606	0.102499546863809\\
0.181818181818182	0.102495922759815\\
0.303030303030303	0.102488679806433\\
0.424242424242424	0.102477828505251\\
0.545454545454546	0.102463384589619\\
0.666666666666667	0.102445369001835\\
0.787878787878788	0.102423807862777\\
0.909090909090909	0.102398732434036\\
1.03030303030303	0.102370179072588\\
1.15151515151515	0.102338189178075\\
1.27272727272727	0.102302809132787\\
1.39393939393939	0.102264090234408\\
1.51515151515152	0.102222088621637\\
1.63636363636364	0.102176865192799\\
1.75757575757576	0.102128485517542\\
1.87878787878788	0.10207701974177\\
2	0.102022542485937\\
2.12121212121212	0.101965132736857\\
2.24242424242424	0.101904873733177\\
2.36363636363636	0.101841852844691\\
2.48484848484848	0.101776161445663\\
2.60606060606061	0.101707894782342\\
2.72727272727273	0.101637151834863\\
2.84848484848485	0.101564035173737\\
2.96969696969697	0.101488650811132\\
3.09090909090909	0.101411108047167\\
3.21212121212121	0.101331519311437\\
3.33333333333333	0.10125\\
3.45454545454546	0.101166668308064\\
3.57575757575758	0.101081645058616\\
3.6969696969697	0.10099505352724\\
3.81818181818182	0.100907019263377\\
3.93939393939394	0.100817669908294\\
4.06060606060606	0.100727135010013\\
4.18181818181818	0.100635545835484\\
4.3030303030303	0.100543035180251\\
4.42424242424242	0.100449737175922\\
4.54545454545454	0.100355787095683\\
4.66666666666667	0.100261321158169\\
4.78787878787879	0.100166476329958\\
4.90909090909091	0.100071390126984\\
5.03030303030303	0.1\\
5.15151515151515	0.1\\
5.27272727272727	0.1\\
5.39393939393939	0.1\\
5.51515151515152	0.1\\
5.63636363636364	0.1\\
5.75757575757576	0.1\\
5.87878787878788	0.1\\
6	0.1\\
};

\addplot [color=red]
  table[row sep=crcr]{%
-0.00153005984029793	0.102798944652196\\
-0.785148742456522	0.102667107559949\\
-1.30107876308287	0.102535609177368\\
-1.68848549942323	0.102404448200831\\
-2.02209249144821	0.102273623333401\\
-2.31958400780897	0.102143133284774\\
-2.55515006926237	0.102012976771247\\
-2.77918578919583	0.101883152515667\\
-2.98997620014521	0.101753659247397\\
-3.19100074177301	0.101624495702268\\
-3.38395048570807	0.101495660622541\\
-3.56657794928122	0.101367152756868\\
-3.74592669871445	0.101238970860245\\
-3.91868503648166	0.101111113693981\\
-4.08788248040528	0.100983580025649\\
-4.25547718939068	0.100856368629053\\
-4.42289441889939	0.100729478284186\\
-4.59485645444949	0.10060290777719\\
-4.78493147158005	0.100476655900321\\
-5.05787095937147	0.100350721451904\\
};

\addplot [color=red]
  table[row sep=crcr]{%
0.000575362637826596	0.102798944652196\\
0.785148742456547	0.102667107559949\\
1.30107876308287	0.102535609177368\\
1.68848549942323	0.102404448200831\\
2.02209249144821	0.102273623333401\\
2.31958400780897	0.102143133284774\\
2.55515006926237	0.102012976771247\\
2.77918578919583	0.101883152515667\\
2.98997620014521	0.101753659247397\\
3.19100074177301	0.101624495702268\\
3.38395048570807	0.101495660622541\\
3.56657794928122	0.101367152756868\\
3.74592669871445	0.101238970860245\\
3.91868503648166	0.101111113693981\\
4.08788248040528	0.100983580025649\\
4.25547718939068	0.100856368629053\\
4.42289441889939	0.100729478284186\\
4.59485645444949	0.10060290777719\\
4.78493147158004	0.100476655900321\\
5.05787095937147	0.100350721451904\\
};

\addplot [color=red]
  table[row sep=crcr]{%
6	0.100350721451904\\
5.05787095937147	0.100350721451904\\
};

\addplot [color=red]
  table[row sep=crcr]{%
-6	0.100350721451904\\
-5.05787095937147	0.100350721451904\\
};

\addplot [color=black]
  table[row sep=crcr]{%
-6	0.093\\
6 0.093\\
};
\end{axis}
\end{tikzpicture}%
} \hspace{5mm}\scalebox{0.8}{\input{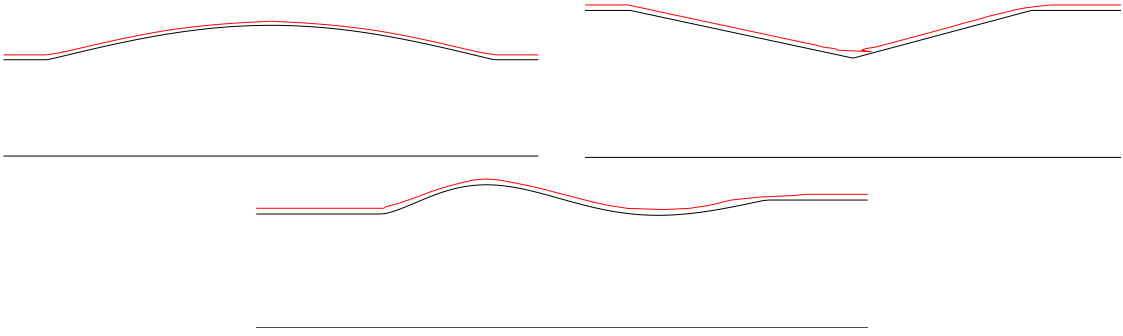}}\\[2mm]
\scalebox{0.8}{
%
%
\begin{tikzpicture}

\begin{axis}[%
width=4in,
height=1in,
at={(0in,0in)},
scale only axis,
xmin=-6,
xmax=6,
ymin=0.094,
ymax=0.102,
hide axis,
axis background/.style={fill=white},
legend style={legend cell align=left, align=left, draw=white!15!black}
]
\addplot [color=black]
  table[row sep=crcr]{%
-6	0.0999616366371927\\
-5.87878787878788	0.0999616366371927\\
-5.75757575757576	0.0999616366371927\\
-5.63636363636364	0.0999616366371927\\
-5.51515151515152	0.0999616366371927\\
-5.39393939393939	0.0999616366371927\\
-5.27272727272727	0.0999616366371927\\
-5.15151515151515	0.0999616366371927\\
-5.03030303030303	0.0999616366371927\\
-4.90909090909091	0.0999616366371927\\
-4.78787878787879	0.0999616366371927\\
-4.66666666666667	0.0999616366371927\\
-4.54545454545454	0.0999616366371927\\
-4.42424242424242	0.0999616366371927\\
-4.3030303030303	0.0999616366371927\\
-4.18181818181818	0.0999616366371927\\
-4.06060606060606	0.0999616366371927\\
-3.93939393939394	0.0999616366371927\\
-3.81818181818182	0.0999616366371927\\
-3.6969696969697	0.0999616366371927\\
-3.57575757575758	0.0999616366371927\\
-3.45454545454545	0.0999888303537817\\
-3.33333333333333	0.100078539241176\\
-3.21212121212121	0.10018836423498\\
-3.09090909090909	0.1003124149704\\
-2.96969696969697	0.100445246697379\\
-2.84848484848485	0.100581938921968\\
-2.72727272727273	0.100718138916011\\
-2.60606060606061	0.100850079043653\\
-2.48484848484848	0.100974574641194\\
-2.36363636363636	0.101089007594362\\
-2.24242424242424	0.101191299583359\\
-2.12121212121212	0.101279878083132\\
-2	0.101353637531035\\
-1.87878787878788	0.101411897549878\\
-1.75757575757576	0.101454359702865\\
-1.63636363636364	0.101481063930527\\
-1.51515151515152	0.101492345558981\\
-1.39393939393939	0.101488793559165\\
-1.27272727272727	0.10147121056736\\
-1.15151515151515	0.101440575040186\\
-1.03030303030303	0.101398005806045\\
-0.909090909090909	0.101344729184849\\
-0.787878787878788	0.101282048774839\\
-0.666666666666667	0.101211317946531\\
-0.545454545454546	0.101133915036576\\
-0.424242424242424	0.101051221196844\\
-0.303030303030303	0.100964600824526\\
-0.181818181818182	0.10087538447624\\
-0.0606060606060606	0.100784854151878\\
0.0606060606060606	0.100694230821302\\
0.181818181818182	0.100604664058278\\
0.303030303030303	0.100517223640416\\
0.424242424242424	0.100432892971053\\
0.545454545454546	0.100352564178248\\
0.666666666666667	0.100277034747143\\
0.787878787878788	0.10020700554449\\
0.909090909090909	0.100143080097805\\
1.03030303030303	0.100085764996273\\
1.15151515151515	0.100035471285829\\
1.27272727272727	0.0999925167367554\\
1.39393939393939	0.0999571288684098\\
1.51515151515152	0.0999294486222348\\
1.63636363636364	0.0999095345809337\\
1.75757575757576	0.0998973676384635\\
1.87878787878788	0.0998928560322859\\
2	0.0998958406560292\\
2.12121212121212	0.0999061005773121\\
2.24242424242424	0.0999233586919168\\
2.36363636363636	0.0999472874517412\\
2.48484848484848	0.0999775146099796\\
2.60606060606061	0.100013628932761\\
2.72727272727273	0.100055185831986\\
2.84848484848485	0.100101712879365\\
2.96969696969697	0.100152715166607\\
3.09090909090909	0.100207680481423\\
3.21212121212121	0.100266084273386\\
3.33333333333333	0.100327394387834\\
3.45454545454546	0.100391075549818\\
3.57575757575758	0.1004565935837\\
3.6969696969697	0.100523419357251\\
3.81818181818182	0.100591032442201\\
3.93939393939394	0.100658924485919\\
4.06060606060606	0.100692820323028\\
4.18181818181818	0.100692820323028\\
4.3030303030303	0.100692820323028\\
4.42424242424242	0.100692820323028\\
4.54545454545454	0.100692820323028\\
4.66666666666667	0.100692820323028\\
4.78787878787879	0.100692820323028\\
4.90909090909091	0.100692820323028\\
5.03030303030303	0.100692820323028\\
5.15151515151515	0.100692820323028\\
5.27272727272727	0.100692820323028\\
5.39393939393939	0.100692820323028\\
5.51515151515152	0.100692820323028\\
5.63636363636364	0.100692820323028\\
5.75757575757576	0.100692820323028\\
5.87878787878788	0.100692820323028\\
6	0.100692820323028\\
};

\addplot [color=black]
  table[row sep=crcr]{%
-7	0.094\\
7	0.094\\
};

\addplot [color=red]
  table[row sep=crcr]{%
-1.5	0.1018\\
-1.74798000708606	0.101739866115266\\
-1.99937291817646	0.101583614709515\\
-2.2221656085855	0.101427843922559\\
-2.41781535192187	0.101272551540276\\
-2.5883139179328	0.101117735362126\\
-2.73817267520958	0.10096339320104\\
-2.89406349312197	0.100809522883325\\
-3.058862406904	0.100656122248556\\
-3.22307385327289	0.100503189149478\\
-3.44788525280986	0.100350721451904\\
-3.5	0.100261636637193\\
-7	0.100261636637193\\
};

\addplot [color=red]
  table[row sep=crcr]{%
-1.5	0.1018\\
-1.23464231455301	0.101739866115266\\
-0.893301684596608	0.101572731409822\\
-0.593642809347912	0.101406146548572\\
-0.331580651325293	0.101240108822634\\
-0.10102240857888	0.101074615540892\\
0.140360789282997	0.100909664029848\\
0.374855688060956	0.100745251633479\\
0.615127330727862	0.100581375713099\\
0.899145936852404	0.100418033647208\\
1.29091172163694	0.100255222831365\\
2	0.100192828716468\\
};

\addplot [color=red]
  table[row sep=crcr]{%
6	0.100992820323028\\
4.8	0.100992820323028\\
4.5267397659487	0.100927567379558\\
3.90082370365366	0.100852416096974\\
3.61795513976012	0.100777376980484\\
3.34643663009154	0.100702449779156\\
3.20927722450648	0.100627634242808\\
3.10350721439954	0.100552930122002\\
2.99187701550717	0.100478337168042\\
2.85584665898262	0.100403855132973\\
2.68045791931868	0.100329483769575\\
2.4888524419966	0.100255222831365\\
2	0.100192828716468\\
};

\end{axis}
\end{tikzpicture}%
}
\end{center}
\caption{\label{recgene} Reconstruction of three different general profiles. In black, the initial shape of $\Omega_6$, and in red, the reconstruction, slightly shifted for comparison purposes. In each case, $K=K_i$ is defined in \eqref{Ki2}, $h=h_i$ is defined in \eqref{h5}, \eqref{h6}, \eqref{h7}, and sources of \eqref{eqmatlab} are defined in \eqref{f5}, \eqref{f6}, \eqref{f7}. We also compute the relative error $\Vert h-h^{\app}\Vert_\infty/h_{\max}$ between the exact reconstruction and the approximation. Top left: $i=5$, $\Vert h-h^{\app}\Vert_\infty/h_{\max}=0.57\%$. Top right: $i=6$, $\Vert h-h^{\app}\Vert_\infty/h_{\max}=0.81\%$. Bottom: $i=7$, $\Vert h-h^{\app}\Vert_\infty/h_{\max}=0.97\%$. }
\end{figure}

To conclude, we have provided in this section a method to reconstruct slowly varying width defects given one side boundary measurements of the wavefield at local resonant frequencies. This method works for every type of source and does not require any \textit{a priori} information on the source except the fact that it is located away from the defect. This reconstruction method is very sensitive to small defects and works numerically to reconstruct dilatations or compressions of the waveguide. 

\section{Conclusion}

In this article, we have used the study of the forward problem in slowly varying waveguide presented in \cite{bonnetier2} and the approximation of the solutions as combination of Airy functions to develop a new inverse method to reconstruct the width of slowly varying waveguides. Given wavefield measurements at the surface of the waveguide for different locally resonant frequencies, we reconstruct the associated locally resonant points which lead to a good approximation of the width of the waveguide.

One main advantage of this new method is that is does not require any \emph{a priori} information on the sources, and we believe that it could be applied to develop new non destructive passive monitoring methods. Moreover, using locally resonant modes, this method can detect small variations of the width with a high sensibility. More importantly, when we compare this new method with the usual back scattering method developed for instance in \cite{bonnetier1}, we notice that this new method seems a lot more precise: while the best relative reconstruction errors are of the order of $8\%$ in \cite{bonnetier1}, we reach in this papers relative errors of the order of less than $1\%$. Even if measurements are not taken in the same area, this improvement of the reconstruction precision must be underlined. 

We believe that this work could be extended to elastic waveguides in two dimensions, using the modal decomposition in Lamb modes presented in \cite{pagneux2}. After generalizing it to elastic waveguide, we plan in future works to test the method developed in this article on experimental data to see if the good numerical results obtained using data generated by finite element methods can be reproduced. Indeed, different physical experiments have already been conducted to recover width defects using locally resonant frequencies (see \cite{balogun1,ces1}) and we hope that the present reconstruction method could both justify and improve these width reconstructions.

\section*{Appendix A: Identification of $\supp(h)$, $k_{\min}$ and $k_{\max}$}

Giving a compactly perturbed waveguide $\Omega$, we describe here how one side boundary measurements enable to approximate very precisely the quantities $\supp(h)$, $k_{\min}$ and $k_{\max}$. The article \cite{bourgeois1} mentions that the problem \eqref{eqmatlab} is not well-defined when $k_n(x)=0$ in a non-trivial interval, which especially happens when $k=n\pi/h_{\min}$ or $k=n\pi/h_{\max}$. Numerically, this results in an explosion of the solution when $k$ tends to $n\pi/h_{\min}$ (resp. $n\pi/h_{\max}$) with a source term located in the area where $h(x)=h_{\min}$ (resp. $h(x)=h_{\max}$). This behavior is illustrated in Figure \ref{explosion}. 

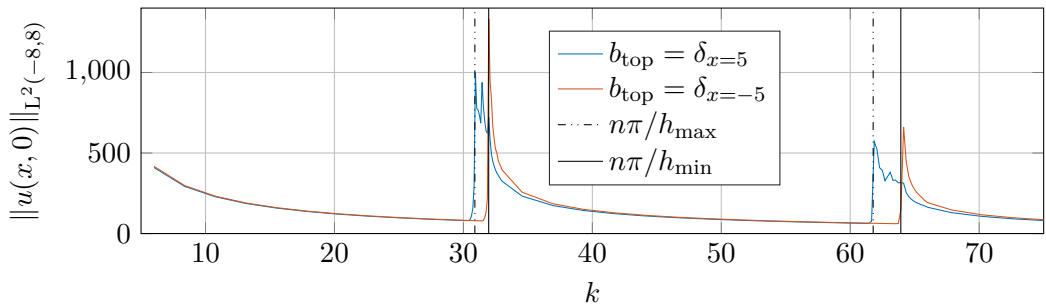
\begin{figure}[h]
\begin{center}
%
%
\definecolor{mycolor1}{rgb}{0.00000,0.44700,0.74100}%
\definecolor{mycolor2}{rgb}{0.85000,0.32500,0.09800}%
\begin{tikzpicture}

\begin{axis}[%
width=12cm,
height=3cm,
at={(0in,0in)},
scale only axis,
xlabel={$k$},
ylabel={$\Vert u(x,0)\Vert_{\text{L}^2(-8,8)}$},
xmin=5,
grid=major,
xmax=75,
ymin=0,
ymax=1400,
axis background/.style={fill=white},
legend style={at={(8.5cm,2.7cm)},legend cell align=left, align=left, draw=white!15!black}
]
\addplot [color=mycolor1]
  table[row sep=crcr]{%
6	410.245138762319\\
8.37931034482759	292.967205367403\\
10.7586206896552	228.394659208477\\
13.1379310344828	186.681966481227\\
15.5172413793103	158.056947753752\\
17.8965517241379	136.816202145284\\
20.2758620689655	120.684274659858\\
22.6551724137931	107.830082102251\\
25.0344827586207	97.4826672528024\\
27.4137931034483	88.8749821547431\\
29.7931034482759	81.7786237711394\\
30	81.258276741279\\
30.1034482758621	81.0217273192172\\
30.2068965517241	80.8192270415955\\
30.3103448275862	80.7149319928121\\
30.4137931034483	81.0361808252224\\
30.5172413793103	86.5956121981685\\
30.6206896551724	107.656256907935\\
30.7241379310345	151.390173055551\\
30.8275862068966	290.980869676186\\
30.9310344827586	997.918142378506\\
31.0344827586207	776.985489427848\\
31.1379310344828	767.007117709666\\
31.2413793103448	734.939558861001\\
31.3448275862069	685.437802785247\\
31.448275862069	938.326857508582\\
31.551724137931	786.852117337351\\
31.6551724137931	696.970238208881\\
31.7586206896552	636.257381220773\\
31.8620689655172	623.945818851526\\
31.9655172413793	720.699416155504\\
32.0689655172414	582.502411397012\\
32.1724137931034	494.074055903955\\
32.1724137931034	494.074055903955\\
32.2758620689655	449.2967518857\\
32.3793103448276	420.883211672971\\
32.4827586206897	395.985318088254\\
32.5862068965517	377.627611991647\\
32.6896551724138	363.308504622313\\
32.7931034482759	349.929754518736\\
32.8965517241379	336.634893010215\\
33	324.954963802868\\
34.551724137931	232.618943951631\\
36.9310344827586	173.937761568224\\
39.3103448275862	143.384741342182\\
41.6896551724138	123.597726853387\\
44.0689655172414	109.506334390173\\
46.448275862069	98.8694410874196\\
48.8275862068966	90.4245504418973\\
51.2068965517241	83.1753371885913\\
53.5862068965517	77.2631193482921\\
55.9655172413793	72.2097686658713\\
58.3448275862069	68.0106981873192\\
60	65.1117454649222\\
60.2068965517241	64.7869198058239\\
60.4137931034483	64.4719137956091\\
60.6206896551724	64.164285272805\\
60.7241379310345	64.016558747165\\
60.8275862068966	63.871568318067\\
61.0344827586207	63.6002947250793\\
61.2413793103448	63.3665795188129\\
61.448275862069	63.3510305496846\\
61.6551724137931	76.8542605276116\\
61.8620689655172	570.901610344758\\
62.0689655172414	523.267860178151\\
62.2758620689655	408.532154144122\\
62.4827586206897	393.76490030418\\
62.6896551724138	328.537244014766\\
62.8965517241379	354.171775639734\\
63.1034482758621	379.984023476524\\
63.1034482758621	379.984023476524\\
63.3103448275862	330.883057724956\\
63.5172413793103	331.243458885134\\
63.7241379310345	315.378664184825\\
63.9310344827586	316.444150550194\\
64.1379310344828	312.935187849312\\
64.3448275862069	253.785756207483\\
64.551724137931	227.960356880653\\
64.7586206896552	211.582614645311\\
64.9655172413793	199.928963036967\\
65.1724137931034	190.831594460409\\
65.3793103448276	183.620520910992\\
65.4827586206897	179.160470828435\\
65.5862068965517	177.144852888643\\
65.7931034482759	170.885122926828\\
66	163.254497517519\\
67.8620689655172	129.689871066379\\
70.2413793103448	106.362993679447\\
72.6206896551724	91.4595586510619\\
75	80.8603858835082\\
};
\addlegendentry{$b_\top=\delta_{x=5}$}

\addplot [color=mycolor2]
  table[row sep=crcr]{%
6	417.299041468994\\
8.37931034482759	298.001091907732\\
10.7586206896552	232.314403744097\\
13.1379310344828	189.879983425277\\
15.5172413793103	160.75879713723\\
17.8965517241379	139.148177746629\\
20.2758620689655	122.734084435149\\
22.6551724137931	109.651316722234\\
25.0344827586207	99.1153777548369\\
27.4137931034483	90.3394231908115\\
29.7931034482759	83.0065133712697\\
30	82.4204967346873\\
30.1034482758621	82.1296294914368\\
30.2068965517241	81.8376864963243\\
30.3103448275862	81.5504381498211\\
30.4137931034483	81.2628588888073\\
30.5172413793103	80.9817577328703\\
30.6206896551724	80.7014640488677\\
30.7241379310345	80.4280988622646\\
30.8275862068966	80.1564941496669\\
30.9310344827586	79.8907698270402\\
31.0344827586207	79.6282008581382\\
31.1379310344828	79.3710251609212\\
31.2413793103448	79.1218710375566\\
31.3448275862069	78.889565112494\\
31.448275862069	78.7378457230911\\
31.551724137931	81.7580555999544\\
31.6551724137931	100.08966590555\\
31.7586206896552	134.012895243036\\
31.8620689655172	218.259786052287\\
31.9655172413793	1334.10185819293\\
32.0689655172414	944.098973008468\\
32.1724137931034	840.036197960668\\
32.1724137931034	840.036197960668\\
32.2758620689655	668.572659473977\\
32.3793103448276	591.036320939775\\
32.4827586206897	527.168019857259\\
32.5862068965517	502.634255227743\\
32.6896551724138	459.506778371188\\
32.7931034482759	433.134667642224\\
32.8965517241379	416.922937897623\\
33	395.201527455345\\
34.551724137931	256.911053199194\\
36.9310344827586	185.189044116559\\
39.3103448275862	150.31541885455\\
41.6896551724138	128.812766426102\\
44.0689655172414	113.513600038883\\
46.448275862069	101.997476882375\\
48.8275862068966	92.9772658540498\\
51.2068965517241	85.7042003742583\\
53.5862068965517	79.3683374746706\\
55.9655172413793	74.0474846854562\\
58.3448275862069	69.4964246893368\\
60	66.793602169312\\
60.2068965517241	66.448413449259\\
60.4137931034483	66.0997833344706\\
60.6206896551724	65.7664836630519\\
60.7241379310345	65.5996922061466\\
60.8275862068966	65.4363189215697\\
61.0344827586207	65.1022451860952\\
61.2413793103448	64.7725679330173\\
61.448275862069	64.4403237439952\\
61.6551724137931	64.111193923064\\
61.8620689655172	63.7767571942917\\
62.0689655172414	63.4553476045515\\
62.2758620689655	63.1312423879686\\
62.4827586206897	62.7994909374348\\
62.6896551724138	62.4752313870747\\
62.8965517241379	62.1465107346859\\
63.1034482758621	61.8165855465304\\
63.1034482758621	61.8165855465304\\
63.3103448275862	61.4956654630587\\
63.5172413793103	61.2090467576091\\
63.7241379310345	62.1735580650928\\
63.9310344827586	147.764125777794\\
64.1379310344828	661.548372831516\\
64.3448275862069	443.148475331862\\
64.551724137931	351.604865907013\\
64.7586206896552	307.615706027244\\
64.9655172413793	276.237214280065\\
65.1724137931034	253.30393989544\\
65.3793103448276	236.634874680944\\
65.4827586206897	228.925986223694\\
65.5862068965517	221.243142053035\\
65.7931034482759	203.860998070681\\
66	192.337342169935\\
67.8620689655172	146.400226711883\\
70.2413793103448	116.254341035746\\
72.6206896551724	98.080019190676\\
75	85.9578703885515\\
};
\addlegendentry{$b_\top=\delta_{x=-5}$}

\addplot [color=black, dash dot dot]
  table[row sep=crcr]{%
30.8887583926631	0\\
30.8887583926631	1400\\
};
\addlegendentry{$n\pi/h_{\max}$}

\addplot [color=black]
  table[row sep=crcr]{%
31.961401114926	0\\
31.961401114926	1400\\
};
\addlegendentry{$n\pi/h_{\min}$}

\addplot [color=black]
  table[row sep=crcr]{%
63.922802229852	0\\
63.922802229852	1400\\
};

\addplot [color=black, dash dot dot]
  table[row sep=crcr]{%
61.7775167853263	0\\
61.7775167853263	1400\\
};

\end{axis}
\end{tikzpicture}%
\caption{\label{explosion} $\text{L}^2$-norm of $u(x,0)$ on the interval $(-8,8)$ with respect to $k$ for a source $b_\top=\delta_{x=-5}$ at the left of $\supp(h)$, and a source $b_\top=\delta_{x=5}$ at the right of $\supp(h)$. For comparison purposes, $n\pi/h_{\max}$ and $n\pi/h_{\min}$ are plotted for $n=1$ and $n=2$.}
\end{center}
\end{figure}

Measuring the surface wavefield while $k$ varies and detecting its explosions provides a good approximation of the width at the left and the right of the waveguide. Then, we choose a frequency $k=n\pi/h_{\max}$ or $k=n\pi/h_{\min}$ and we move the sources while measuring the amplitude of the wavefield. It the source is located outside of the support of $h$, the wavefield is supposed to explode, which provides a good approximation of the support of $h$. This behavior is illustrated in Figure \ref{explosionsource}. 

\begin{figure}[h]
\begin{center}
%
%
\definecolor{mycolor1}{rgb}{0.00000,0.44700,0.74100}%
\definecolor{mycolor2}{rgb}{0.85000,0.32500,0.09800}%
\begin{tikzpicture}

\begin{axis}[%
width=12cm,
height=3.5cm,
scale only axis,
xmin=-8,
xmax=8,
ymin=0,
grid=major,
ymax=300,
xlabel={$s$},
ylabel={$\Vert u(x,0)\Vert_{\text{L}^2(-8,8)}$},
axis background/.style={fill=white},
legend style={at={(7cm,3.5cm)},legend cell align=left, align=left, draw=white!15!black}
]
\addplot [color=mycolor1]
  table[row sep=crcr]{%
-8	257.79473896045\\
-7.79746835443038	257.500960002422\\
-7.59493670886076	256.640053798649\\
-7.39240506329114	255.20864285822\\
-7.18987341772152	253.19717629802\\
-6.9873417721519	250.591082753004\\
-6.78481012658228	247.372478656442\\
-6.58227848101266	243.522611646454\\
-6.37974683544304	239.025587384581\\
-6.17721518987342	233.873481081541\\
-5.9746835443038	228.074442203938\\
-5.77215189873418	221.664398499861\\
-5.56962025316456	214.72497486165\\
-5.36708860759494	207.409970982114\\
-5.16455696202532	199.983538603264\\
-4.9620253164557	192.874412364155\\
-4.75949367088608	186.743487601021\\
-4.55696202531646	182.552809090449\\
-4.35443037974684	181.596428945067\\
-4.15189873417722	185.451932960771\\
-3.9493670886076	192.908848652084\\
-3.74683544303797	180.647316671848\\
-3.54430379746835	128.787416692778\\
-3.34177215189873	93.8751631259312\\
-3.13924050632911	140.162741514071\\
-2.93670886075949	151.900683733612\\
-2.73417721518987	100.399992683896\\
-2.53164556962025	107.700870696516\\
-2.32911392405063	145.667939201803\\
-2.12658227848101	107.161155475572\\
-1.92405063291139	98.9701092434727\\
-1.72151898734177	139.111727369496\\
-1.51898734177215	101.170093225401\\
-1.31645569620253	103.25832096072\\
-1.11392405063291	133.371318206171\\
-0.911392405063291	89.0931722385449\\
-0.708860759493671	115.86459466311\\
-0.506329113924051	120.464184969915\\
-0.30379746835443	86.0550876936622\\
-0.10126582278481	126.639464731257\\
0.10126582278481	96.692050528981\\
0.30379746835443	105.136428869084\\
0.506329113924051	118.700583885827\\
0.708860759493671	85.6063192034788\\
0.911392405063291	122.392182884456\\
1.11392405063291	91.0678523587159\\
1.31645569620253	109.587342353088\\
1.51898734177215	107.265731199676\\
1.72151898734177	92.9866151363349\\
1.92405063291139	116.959212698144\\
2.12658227848101	85.0516088977993\\
2.32911392405063	116.854370055179\\
2.53164556962025	88.6377678631456\\
2.73417721518987	109.987302857864\\
2.93670886075949	97.2133889655881\\
3.13924050632911	100.77219672181\\
3.34177215189873	104.832892302709\\
3.54430379746835	92.7901429013525\\
3.74683544303797	109.378195154259\\
3.9493670886076	87.8618807895942\\
4.15189873417722	111.109002968433\\
4.35443037974684	86.4762990877706\\
4.55696202531646	111.185308310165\\
4.75949367088608	87.6221334539322\\
4.9620253164557	109.383363705854\\
5.16455696202532	90.8082248599516\\
5.36708860759494	106.108087810291\\
5.56962025316456	95.1350549880584\\
5.77215189873418	101.966055314798\\
5.9746835443038	99.6368774579187\\
6.17721518987342	97.7054425163506\\
6.37974683544304	103.509473415895\\
6.58227848101266	94.1107490041678\\
6.78481012658228	106.205679839587\\
6.9873417721519	91.8424001924636\\
7.18987341772152	107.452628560087\\
7.39240506329114	91.2554319424553\\
7.59493670886076	107.233394127504\\
7.79746835443038	92.2904892448335\\
8	105.755332608686\\
};
\addlegendentry{$k=n\pi/h_{\min}$}

\addplot [color=mycolor2]
  table[row sep=crcr]{%
-8	4.41120387792682\\
-7.79746835443038	4.41120436307943\\
-7.59493670886076	4.41120416620812\\
-7.39240506329114	4.41120433869579\\
-7.18987341772152	4.41120432276895\\
-6.9873417721519	4.4112038173335\\
-6.78481012658228	4.41120355254779\\
-6.58227848101266	4.41120361147306\\
-6.37974683544304	4.41120398956229\\
-6.17721518987342	4.41120430752882\\
-5.9746835443038	4.41120418130879\\
-5.77215189873418	4.41120446292364\\
-5.56962025316456	4.411204173799\\
-5.36708860759494	4.41120384265491\\
-5.16455696202532	4.41120511472569\\
-4.9620253164557	4.41123403664579\\
-4.75949367088608	4.4117807514606\\
-4.55696202531646	4.42206358221671\\
-4.35443037974684	4.61234260757265\\
-4.15189873417722	7.48488187777834\\
-3.9493670886076	13.0454014436152\\
-3.74683544303797	7.09325265973953\\
-3.54430379746835	5.65155853743271\\
-3.34177215189873	5.77984755460665\\
-3.13924050632911	6.03791754765899\\
-2.93670886075949	6.30234669592664\\
-2.73417721518987	6.56768253315139\\
-2.53164556962025	6.83555978448673\\
-2.32911392405063	7.10898201353008\\
-2.12658227848101	7.39011125449505\\
-1.92405063291139	7.68080522972715\\
-1.72151898734177	7.98331247415601\\
-1.51898734177215	8.29914695149029\\
-1.31645569620253	8.63101723211177\\
-1.11392405063291	8.98052083989453\\
-0.911392405063291	9.35060666608337\\
-0.708860759493671	9.74373659934539\\
-0.506329113924051	10.1635165620467\\
-0.30379746835443	10.613447719317\\
-0.10126582278481	11.0983178639703\\
0.10126582278481	11.6238079324692\\
0.30379746835443	12.1972033237013\\
0.506329113924051	12.8209649269128\\
0.708860759493671	13.5106536730882\\
0.911392405063291	14.2773342523045\\
1.11392405063291	15.1377764436077\\
1.31645569620253	16.1144438547319\\
1.51898734177215	17.2416277779749\\
1.72151898734177	18.5720373425189\\
1.92405063291139	20.1949823330564\\
2.12658227848101	22.2632118649963\\
2.32911392405063	25.0359693180714\\
2.53164556962025	28.9232726730064\\
2.73417721518987	34.5021612189022\\
2.93670886075949	42.4599111324391\\
3.13924050632911	53.44743613183\\
3.34177215189873	67.8810487059051\\
3.54430379746835	85.7168644194464\\
3.74683544303797	106.218356143713\\
3.9493670886076	127.742886264645\\
4.15189873417722	147.860078471303\\
4.35443037974684	165.752805584987\\
4.55696202531646	181.59611513224\\
4.75949367088608	195.573504570956\\
4.9620253164557	207.865231340894\\
5.16455696202532	218.640020793299\\
5.36708860759494	228.052107035235\\
5.56962025316456	236.24046952568\\
5.77215189873418	243.32906250302\\
5.9746835443038	249.427407130117\\
6.17721518987342	254.631363231928\\
6.37974683544304	259.024104966125\\
6.58227848101266	262.6767685402\\
6.78481012658228	265.649178738063\\
6.9873417721519	267.990520213281\\
7.18987341772152	269.739872213256\\
7.39240506329114	270.926635161313\\
7.59493670886076	271.570836480811\\
7.79746835443038	271.683359567865\\
8	271.266217421248\\
};
\addlegendentry{$k=n\pi/h_{\max}$}

\addplot [color=black, ultra thick]
  table[row sep=crcr]{%
-4 0 \\
-4 300\\
};
\addlegendentry{$\supp(h)$}

\addplot [color=black,  ultra thick]
  table[row sep=crcr]{%
6 0 \\
6 300\\
};

\end{axis}
\end{tikzpicture}%
\caption{\label{explosionsource} $\text{L}^2$-norm of $u(x,0)$ on the interval $(-8,8)$ with respect to the position $s$ of the source $b_\top=\delta_{x=s}$, for a frequency $k=pi/h_{\max}$ and $k=\pi/h_{\min}$. For comparison purposes the support of $h$ is plotted.}
\end{center}
\end{figure}
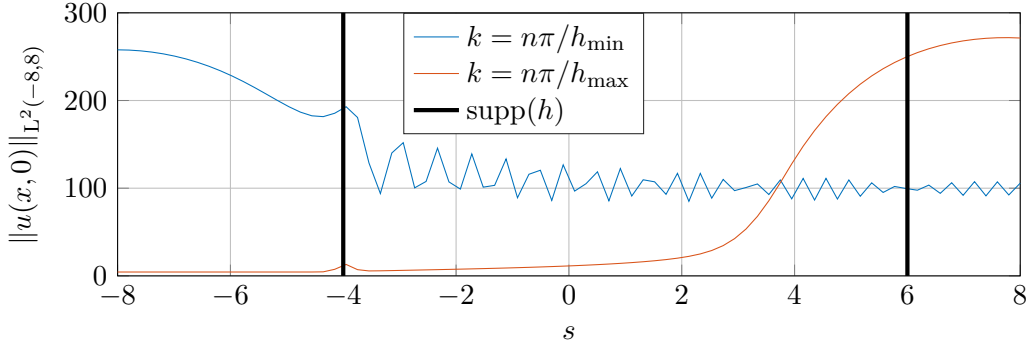

\section*{Appendix B: Proof of Proposition \ref{leastsq}}
\begin{proof}
For every $\g p\in U$, we compute 
\begin{equation}
\nabla J_{\g d}(\g p)=\left(\begin{array}{c} \sum_{i=1}^{n}\Ai(\beta-\alpha t_i)(z\Ai(\beta-\alpha t_i)-d_i)\\
\sum_{i=1}^{n}-zt_i\Ai'(\beta-\alpha t_i)(z\Ai(\beta-\alpha t_i)-d_i) \\
\sum_{i=1}^{n}z\Ai'(\beta-\alpha t_i)(z\Ai(\beta-\alpha t_i)-d_i) \end{array}\right), 
\end{equation}
\begin{equation}
\nabla^2J_{\g d}(\g p)=\sum_{i=1}^{n}M_i(\beta-\alpha t_i), 
\end{equation}
with 
\begin{equation}
M_i=\left(\begin{array}{ccc} \Ai^2 & -t_i\Ai'\times (2z\Ai-d_i) & \Ai'(2z\Ai-d_i) \\ -t_i\times \Ai'\times (2z\Ai-d_i) & \begin{array}{c} zt_i^2\Ai''\times (z\Ai-d_i)\\ \qquad  +z^2t_i^2\times (\Ai')^2\end{array} & \begin{array}{c} -zt_i\times \Ai''(z\Ai-d_i)\\ \qquad  -t_iz^2(\Ai')^2 \end{array}\\\Ai'\times (2z\Ai-d_i) & \begin{array}{c} -zt_i\times \Ai''(z\Ai-d_i)\\ \qquad -t_iz^2\times (\Ai')^2 \end{array} & \begin{array}{c} z\Ai''\times (z\Ai-d_i)\\ \qquad z^2(\Ai')^2 \end{array}\end{array}\right).\end{equation}
For every $h\in \C\times\R^2$, $h\neq 0$, 
\begin{equation}
(\nabla^2J_{\g d}(\g p)h|h)=A(\g p,h)+B(\g p,h), 
\end{equation}
where 
\begin{equation}
A(\g p,h)=\sum_{i=1}^{n} (\Ai(\beta-\alpha t_i)h_1+z\Ai'(\beta-\alpha t_i)(h_3-h_2t_i))^2,
\end{equation}
\begin{equation}
B(\g p,h)=\sum_{i=1}^{n}(z\Ai(\beta-\alpha t_i)-d_i)(2\Ai'(\beta-\alpha t_i)h_1(h_3-h_2t_i)+z\Ai''(\beta-\alpha t_i)(h_3-h_2t_i)^2).
\end{equation}
We want to prove that $A(\g p,h)>0$. To do so, we use the following Lemma. 

\begin{lem}
For every $(h_1,h_2,h_3)\in \C\times \R\times \R$ and $(z,\alpha,\beta)\in U$, the set of zeros of
\begin{equation}
x\mapsto h_1\Ai(\beta-\alpha x)+z\Ai'(\beta-\alpha x)\left(h_3-h_2x\right),
\end{equation} 
is finite, and at most equal to $3\ell+3$ where $\ell$ is the number of zeros of $\Ai'$ on 
\begin{equation}
I:=[\beta_{\min}-t,\beta_{\max}+t], \qquad \text{where} \qquad t=\alpha_{\max}\max(|t_1|,|t_n|).
\end{equation}
\end{lem}

\begin{proof}
We do a change of variable $u=\beta-\alpha x$, and we see that for every $x\in [x_1,x_n]$, $u\in I$. We now look for $u\in I$ satisfying 
\begin{equation}
g_1(u):= h_1\Ai(u)+z\Ai'(u)\left(h_3-\frac{h_2}{\alpha}(\beta-u)\right)=0.
\end{equation}
We notice that if $h_1=0$ then $g_1(u)=0$ if and only if $\Ai'(u)=0$ or $x=h_3/h_2$, which gives $\ell+1$ zeros of $g_1$. Otherwise, we notice that if $\Ai'(u)=0$ then $g_1(u)\neq 0$ since every zero of the Airy function is simple (see \cite{abramowitz1}). It means that there exist $\alpha\in \C$ and $\beta \in \R$ such that
\begin{equation}
g_1(u)=0 \quad \Leftrightarrow \quad \frac{\Ai}{\Ai'}(u)=\alpha u +\beta.
\end{equation} 
We define 
\begin{equation}
g_2(u)=\frac{\Ai}{\Ai'}(u)-\alpha u -\beta, 
\end{equation}
and then 
\begin{equation}
g_2'(u)=\left(1-u\left(\frac{\Ai(u)}{\Ai'(u)}\right)^2\right), \qquad g_2''(u)=\frac{\Ai(u)}{\Ai'(u)}\underbrace{\left(2-2u\left(\frac{\Ai(u)}{\Ai'(u)}\right)^2-\frac{\Ai(u)}{\Ai'(u)}\right)}_{>0}. 
\end{equation}
Between two zeros of $\Ai'$, $\Ai$ vanishes only once, meaning that depending of the value of $\alpha$, $g_2'$ vanishes at most twice, and so depending of the value of $\beta$, $g_2$ vanishes at most three times.
\end{proof}

Back to the proof of Proposition \ref{leastsq}, we now set $n_0=3\ell+3$, and if $n>n_0$ then it shows that $A(\g p,h)>0$. We denote $\lambda_1(\g p)=\min_{h\in \R^3, h\neq 0}a(\g p,h)/\Vert h \Vert_2^2$. This function is continuous on a subset $U_1\subset U$, and we denote by $m$ the minimum of $\lambda_1$ on $U_1$. We also notice that
\begin{equation}
|b(\g p,h)|\leq \Vert z\Ai(\beta-\alpha X) -\g d\Vert_1(\Vert \Ai'\Vert _\infty+2z\Vert \Ai''\Vert_\infty)(1+\Vert \g t\Vert_\infty)\Vert h\Vert_2^2.
\end{equation}
We see that 
\begin{equation}
\Vert z\Ai(\beta-\alpha X) -\g d\Vert_1=\Vert F(\g p)-\g d\Vert_1\leq \Vert F(\g p)-F(\g p_0)\Vert_1+\Vert \g d_0-\g d\Vert_1.
\end{equation}
There exists a constant $M>0$ depending on $U$ such that for every $\g p\in U$, 
\begin{equation}
\Vert F(\g p)-F(\g p_0)\Vert_1\leq M\Vert \g p-\g p_0\Vert_1.
\end{equation}
We define $U_2=U_1\cap B_1(p_0,\eps_1/M)$ where 
\begin{equation}
\eps_1:=\frac{m}{4(\Vert \Ai'\Vert _\infty+2z_{\max}\Vert \Ai''\Vert_\infty)(1+\max(|x_1|,|x_n|))}.
\end{equation}
It follows that for every $\g p\in U_2$ and $\g d\in B_1(\g d_0,\eps_1)$,
\begin{equation}
\Vert F(\g p)-F(\g p_0)\Vert_1\leq \eps_1, \quad |b(\g p,h)|\leq \frac{m}{2}\Vert h \Vert_2^2.
\end{equation}
The operator $J_{\g d}$ is then strictly convex on $U_2$ since 
\begin{equation}
\Vert \nabla^2J_{\g d}(\g p)\Vert_2\geq \min_{h\in \C\times \R^2, h\neq 0}\frac{| A(\g p,h)|-|B(\g p,h)|}{\Vert h \Vert_2^2}\geq \frac{m}{2}>0
\end{equation}
We now need to prove that the minimum of $J_{\g d}$ is located inside of $U_2$ and not on its boundary. To do so, we look for a point $\g p\in U_2$ such that $\nabla J_{\g d}(\g p)=0$. We already know that $\nabla J_{\g d_0}(\g p_0)=0$. Using the implicit function theorem, there exists an open set $V\subset \R^{n}$ containing $\g d_0$ such that there exists a unique continuously differentiable function $G_1: V\rightarrow \C\times \R\times \R$ such that $\nabla F_{\g d}(G_1(\g d))=0$. We define $U'=G_1(V)\cap U_2$ and there exists $\eps>0$ such that $B_2(\g d_0,\eps)\subset G_1^{-1}(U')\cap B_1(\g d_0,\eps_1)$. It shows that the application 
\begin{equation}
G:\left\{\begin{array}{rcl} B_2(\g d_0,\eps) & \rightarrow & U' \\ \g d& \mapsto& \text{argmin}_{\g p\in U'}J_{\g d}(\g p)\end{array}\right. ,
\end{equation}
is well-defined and continuously differentiable. Moreover, we also know that  
\begin{equation}
\partial_{d_j}G(\g d)=-[\nabla^2J_{\g d}(G(\g d))]^{-1}[\partial_{d_j}\nabla J_{\g d}(G(\g d))].
\end{equation}
We denote $\g p_\text{LS}=G(\g d)$, and  
\begin{equation}
\Vert \partial_{d_j}\nabla J_{\g d}(\g p_\text{LS})\Vert_2=\left\Vert\left(\begin{array}{c} -\Ai(\beta_\text{LS}-\alpha_\text{LS}t_j) \\
z_\text{LS}t_j\Ai'(\beta_\text{LS}-\alpha_\text{LS}t_j) \\
-z_\text{LS}\Ai'(\beta_\text{LS}-\alpha_\text{LS}t_j)\end{array}\right) \right\Vert_2\leq \sqrt{\Vert \Ai\Vert_\infty^2+\alpha_{\max}\Vert \Ai'\Vert_\infty^2(1+\Vert \g t\Vert_\infty)}:=c_1.
\end{equation}
If follows that 
\begin{equation}
\Vert\g p_\text{LS}-\g p_0\Vert_2=\Vert G(\g d)-G(\g d_0)\Vert_2\leq \frac{2\sqrt{n}c_1}{m}\Vert \g d-\g d_0\Vert_2.
\end{equation}
Finally, 
\begin{equation}
|\Lambda(\g d)-\Lambda(\g d_0)|\leq \frac{1}{\alpha_{\min}}|\beta_\text{LS}-\beta_0|+\frac{\beta_{\max}}{(\alpha_{\min})^2}|\alpha_\text{LS}-\beta_\text{LS}|\leq \left(\frac{1}{\alpha_{\min}}+\frac{\beta_{\max}}{(\alpha_{\min})^2}\right)\frac{2\sqrt{n}c_1}{m}\Vert \g d-\g d_0\Vert_{2}. 
\end{equation}
\end{proof}

\bibliographystyle{abbrv}
\bibliography{biblio}

\begin{thebibliography}{10}

\bibitem{abra2}
L.~Abrahamsson.
\newblock Orthogonal grid generation for two-dimensional ducts.
\newblock {\em Journal of Computational and Applied Mathematics},
  34(3):305--314, 1991.

\bibitem{abra1}
L.~Abrahamsson and H.~O. Kreiss.
\newblock Numerical solution of the coupled mode equations in duct acoustics.
\newblock {\em Journal of Computational Physics}, 111(1):1--14, 1994.

\bibitem{abramowitz1}
M.~Abramowitz and I.~A. Stegun.
\newblock {\em Handbook of Mathematical Functions: With Formulas, Graphs, and
  Mathematical Tables}.
\newblock Applied mathematics series. Dover Publications, 1965.

\bibitem{acosta1}
S.~Acosta, S.~Chow, J.~Taylor, and V.~Villamizar.
\newblock On the multi-frequency inverse source problem in heterogeneous media.
\newblock {\em Inverse Problems}, 28(7):075013, 2012.

\bibitem{balogun1}
O.~Balogun, T.~W. Murray, and C.~Prada.
\newblock Simulation and measurement of the optical excitation of the s1 zero
  group velocity lamb wave resonance in plates.
\newblock {\em Journal of Applied Physics}, 102(6):064914, 2007.

\bibitem{bao2}
G.~Bao and P.~Li.
\newblock Inverse medium scattering problems for electromagnetic waves.
\newblock {\em SIAM Journal on Applied Mathematics}, 65(6):2049--2066, 2005.

\bibitem{bao1}
G.~Bao and F.~Triki.
\newblock Reconstruction of a defect in an open waveguide.
\newblock {\em Science China Mathematics}, 56(12):2539--2548, 2013.

\bibitem{bao3}
G.~Bao and F.~Triki.
\newblock Stability for the multifrequency inverse medium problem.
\newblock {\em Journal of Differential Equations}, 269(9):7106--7128, 2020.

\bibitem{berenger1}
J.~P. Berenger.
\newblock A perfectly matched layer for the absorption of electromagnetic
  waves.
\newblock {\em Journal of Computational Physics}, 114(2):185--200, 1994.

\bibitem{bonnetier2}
E.~Bonnetier, A.~Niclas, L.~Seppecher, and G.~Vial.
\newblock The helmholtz problem in slowly varying waveguides at locally
  resonant frequencies.
\newblock {\em submitted in Wave Motion}, 2022.

\bibitem{bonnetier1}
E.~Bonnetier, A.~Niclas, L.~Seppecher, and G.~Vial.
\newblock {Small defects reconstruction in waveguide from multifrequency
  one-side scattering data}.
\newblock {\em {Inverse Problems and Imaging}}, 16(2):417--450, 2022.

\bibitem{bourgeois1}
L.~Bourgeois and E.~Lun{\'e}ville.
\newblock {The linear sampling method in a waveguide: A modal formulation}.
\newblock {\em {Inverse Problems}}, 24(1), 2008.

\bibitem{ces1}
M.~C{\`e}s.
\newblock {\em {Etude th{\'e}orique et exp{\'e}rimentale des r{\'e}sonances
  m{\'e}caniques locales de modes guid{\'e}s par des structures complexes.}}
\newblock Theses, {Universit{\'e} Paris-Diderot - Paris VII}, 2012.

\bibitem{davis1}
P.~J. Davis and P.~Rabinowitz.
\newblock {\em Methods of Numerical Integration}.
\newblock Academic Press, second edition, 1984.

\bibitem{folguera1}
A.~Folguera and J.~G. Harris.
\newblock Coupled raleigh surface waves in slowly varying elastic waveguide.
\newblock {\em Proceedings of the Royal Society of London. Series A:
  Mathematical, Physical and Engineering Sciences}, 455(1983):917--931, 1999.

\bibitem{honarvar1}
F.~Honarvar, F.~Salehi, V.~Safavi, A.~Mokhtari, and A.~N. Sinclair.
\newblock Ultrasonic monitoring of erosion/corrosion thinning rates in
  industrial piping systems.
\newblock {\em Ultrasonics}, 53(7):1251--1258, 2013.

\bibitem{isakov1}
V.~Isakov and S.~Lu.
\newblock Increasing stability in the inverse source problem with attenuation
  and many frequencies.
\newblock {\em SIAM J. Appl. Math.}, 78:1--18, 2018.

\bibitem{legrand2}
F.~Legrand.
\newblock {\em {Ondes de Lamb et R{\'e}fraction N{\'e}gative}}.
\newblock PhD thesis, {Sorbonne University , UPMC}, 2020.

\bibitem{legrand1}
F.~Legrand, B.~Gérardin, J.~Laurent, C.~Prada, and A.~Aubry.
\newblock Negative refraction of lamb modes: A theoretical study.
\newblock {\em Physical Review B}, 98(21), 2018.

\bibitem{lu1}
Y.~Y. Lu.
\newblock Exact one-way methods for acoustic waveguides.
\newblock {\em Mathematics and Computers in Simulation}, 50(5):377--391, 1999.

\bibitem{mallat1}
S.~Mallat.
\newblock {\em A Wavelet Tour of Signal Processing, Chapter 1: Sparse
  Representations}.
\newblock Academic Press, Boston, third edition, 2009.

\bibitem{norgren1}
M.~Norgren, M.~Dalarsson, and A.~Motevasselian.
\newblock Reconstruction of boundary perturbations in a waveguide.
\newblock {\em 2013 International Symposium on Electromagnetic Theory}, pages
  934--937, 2013.

\bibitem{pagneux2}
V.~Pagneux and A.~Maurel.
\newblock Lamb wave propagation in elastic waveguides with variable thickness.
\newblock {\em Proceedings of the Royal Society A: Mathematical, Physical and
  Engineering Sciences}, 462(2068):1315--1339, 2006.

\bibitem{sini1}
M.~Sini and N.~T. Thanh.
\newblock Inverse acoustic obstacle scattering problems using multifrequency
  measurements.
\newblock {\em Inverse Problems and Imaging}, 6(4):749--773, 2012.

\end{thebibliography}

\end{document}